\documentclass[conference]{IEEEtran}
\IEEEoverridecommandlockouts
\usepackage{amsmath,amssymb,amsfonts}
\usepackage{graphicx}
\usepackage{textcomp}
\usepackage{xcolor}
\usepackage{tikz}
\usepackage{amsthm}
\usepackage{comment}
\usepackage{csquotes}
\usepackage{enumitem}
\usepackage{filecontents}
\usepackage{algorithm}
\usepackage{algorithmicx}
\usepackage{algpseudocode}
\usepackage[nolist]{acronym}
\usepackage{xspace}
\usepackage{booktabs}
\usepackage{xcolor}
\usepackage{threeparttable}
\usepackage{pifont}
\usepackage{hyperref}
\usepackage{multirow}
\usepackage[style=numeric,maxbibnames=99]{biblatex}
\newcommand{\Z}{\mathbb{Z}}
\newtheorem{theorem}{Theorem}
\renewcommand\vec{\mathbf}

\newcommand{\circled}[2][]{
  \tikz[baseline=(char.base)]{
    \node[shape=circle,inner sep=0.8pt,fill=black]
    (char) {\phantom{\ifblank{#1}{#2}{#1}}};
    \node[text=white] at (char.center) {\makebox[0pt][c]{\textbf#2}};}\xspace}
\robustify{\circled}

\def\orcid#1{\kern.08em\href{https://orcid.org/#1}{\protect\includegraphics[keepaspectratio,width=0.7em]{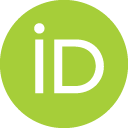}}}

\newcommand{\xilinx}[1]{\textcolor{black}{#1}}
\newcommand{\intel}[1]{\textcolor{black}{#1}}
\newcommand{\lattice}[1]{\textcolor{black}{#1}}
\newcommand{\microsemi}[1]{\textcolor{black}{#1}}
\newcommand{\cadence}[1]{\textcolor{black}{#1}}
\newcommand{\siemens}[1]{\textcolor{black}{#1}}

\definecolor{darkgreen}{RGB}{6, 156, 47}
\definecolor{darkred}{RGB}{240, 2, 2}
\newcommand{\cmark}{{\color{darkgreen}\ding{51}}}

\newcommand{\gbar}{{\color{lightgray}-}}

\usepackage{soul}

\begin{document}

\title{How Not to Protect Your IP -- An Industry-Wide Break of IEEE~1735 Implementations
\thanks{Funded by the Deutsche Forschungsgemeinschaft (DFG, German Research Foundation) under Germany´s Excellence Strategy - EXC 2092 CASA - 390781972 and ERC grant 695022.}
}


\author{\IEEEauthorblockN{
Julian Speith\IEEEauthorrefmark{1}\IEEEauthorrefmark{2}\orcid{0000-0002-8408-8518},
Florian Schweins\IEEEauthorrefmark{2}\orcid{0000-0003-0605-5937},
Maik Ender\IEEEauthorrefmark{1}\IEEEauthorrefmark{2}\orcid{0000-0002-0685-2541},
Marc Fyrbiak\IEEEauthorrefmark{1}\orcid{0000-0002-4266-7108},
Alexander May\IEEEauthorrefmark{2}\orcid{0000-0001-5965-5675},
Christof Paar\IEEEauthorrefmark{1}\IEEEauthorrefmark{2}\orcid{0000-0001-8681-2277}}
\IEEEauthorblockA{\IEEEauthorrefmark{1}Max Planck Institute for Security and Privacy, Bochum, Germany,\\
\href{mailto:julian.speith@mpi-sp.org}{julian.speith@mpi-sp.org}, 
\href{mailto:maik.ender@mpi-sp.org}{maik.ender@mpi-sp.org}, 
\href{mailto:marc.fyrbiak@mpi-sp.org}{marc.fyrbiak@mpi-sp.org}, 
\href{mailto:christof.paar@mpi-sp.org}{christof.paar@mpi-sp.org}
}
\IEEEauthorblockA{\IEEEauthorrefmark{2}Horst Görtz Institute for IT Security, Ruhr University Bochum, Bochum, Germany,\\
\href{mailto:florian.schweins@rub.de}{florian.schweins@rub.de}, 
\href{mailto:alex.may@rub.de}{alex.may@rub.de}
}}

\maketitle
\thispagestyle{plain}
\pagestyle{plain}

\begin{abstract}
Modern hardware systems are composed of a variety of third-party \ac{IP} cores to implement their overall functionality. 
Since hardware design is a globalized process involving various (untrusted) stakeholders, a secure management of the valuable \acs{IP} between authors and users is inevitable to protect them from unauthorized access and modification. 
To this end, the widely adopted IEEE standard 1735-2014 was created to ensure confidentiality and integrity.

In this paper, we outline structural weaknesses in IEEE~1735 that cannot be fixed with cryptographic solutions (given the contemporary hardware design process) and thus render the standard inherently insecure. 
We practically demonstrate the weaknesses by recovering the private keys of IEEE~1735 implementations from major \ac{EDA} tool vendors, namely \intel{Intel}, \xilinx{Xilinx}, \cadence{Cadence}, \siemens{Siemens}, \microsemi{Microsemi}, and \lattice{Lattice}, while results on a seventh case study are withheld.
As a consequence, we can decrypt, modify, and re-encrypt all allegedly protected \acs{IP} cores designed for the respective tools, thus leading to an industry-wide break.
As part of this analysis, we are the first to publicly disclose three RSA-based white-box schemes that are used in real-world products and present cryptanalytical attacks for all of them, finally resulting in key recovery. 
\end{abstract}


\begin{acronym}
    \acro{API}{Application Programming Interface}
    \acro{ASIC}{Application-Specific Integrated Circuit}
    
    \acro{CRT}{Chinese Remainder Theorem}
    
    \acro{DLL}{Dynamic Link Library}
    \acroplural{DLL}[DLLs]{Dynamic Link Libraries}
    \acro{DRM}{Digital Rights Management}
    \acro{DSP}{Digital Signal Processor}
    
    \acro{ECC}{Elliptic-Curve Cryptography}
    \acro{EDA}{Electronic Design Automation}
    
    \acro{FPGA}{Field-Programmable Gate Array}
    
    \acro{GCD}{Greatest Common Divisor}
    
    \acro{HDL}{Hardware Description Language}
    
    \acro{IC}{Integrated Circuit}
    \acro{IP}{Intellectual Property}
    
    \acro{MATE}{Man-at-the-End}
    \acro{MITM}{Man-in-the-Middle}
    
    \acro{PKI}{Public Key Infrastructure}
    
    \acro{RAM}{Random-Access Memory}
    \acro{RTL}{Register Transfer Level}

    \acro{SoC}{System-on-Chip}
    \acroplural{SoC}[SoCs]{Systems-on-Chip}
\end{acronym}

\section{Introduction}\label{sec:whitebox:intro}
The emerging digital society is based on a myriad of interconnected computers and embedded devices built from \acp{IC}.
Modern \acp{IC} are increasingly realized as \acp{SoC}, i.e., a single chip that incorporates a large number of different functional modules, referred to as \ac{IP} cores. 
In particular, modern \acp{SoC} can consist of hundreds of \ac{IP} cores that implement a wide range of functionalities, from simple periphery, to cryptographic co-processors, neural network accelerators, or even entire CPUs.
Given both the increasing complexity of modern hardware and strict time-to-market requirements~\cite{Bhunia2017}, hardware designers resort to the re-use of well-proven \ac{IP} cores, often from third-party providers.  
However, the use of third-party \ac{IP} cores provides several security challenges from both the \ac{IP} author as well as the \ac{IP} user perspective~\cite{Bhunia2017} due to the globalized hardware design and manufacturing process comprising numerous untrusted stakeholders.
For example, the author wants to safeguard their valuable \ac{IP} so that no other party can infringe their \ac{IP} or its license. 
Simultaneously, the user wants to assure that the purchased \ac{IP} is not modified by another stakeholder. 

\begin{table*}
    \centering
    \label{tab:whitebox:case_study_overview}
    \caption{Overview of our case studies on IEEE~1735 implementations of major \ac{EDA} tools.
        For each tool, we list whether we defeated the encountered software protections (\cmark) or no such protections were present at the time of analysis (\gbar). Additionally, we provide the number of recovered private keys and the approximated time required to perform each case study.}
    \begin{tabular}{lccccl}
        \toprule
        \multirow{2}{*}{\textbf{\ac{EDA} Tool}} & \multicolumn{3}{c}{\textbf{Defeated Software Protections}} & \textbf{\#Recovered} &  \multicolumn{1}{c}{\textbf{Time}}\\
        & Obfuscation & Anti-Debugging & White-Box & \textbf{Private Keys} &  \multicolumn{1}{c}{\textbf{Estimate}}\\
        \midrule
        \intel{Intel Quartus Prime} & \gbar & \gbar & \gbar & \intel{1} & $1$ hour\\
        \cadence{Cadence Xcelium} & \gbar & \gbar & \gbar & \cadence{2} & $1$ hour\\
        \lattice{Lattice Radiant} & \gbar & \gbar & \gbar & \lattice{1} & $3$ hours\\
        \microsemi{Microsemi Libero SoC} & \cmark & \cmark & (\cmark) & \microsemi{1} & $1$ week\\
        \siemens{Siemens ModelSim} & \gbar & \gbar & \cmark & \siemens{3} & $3$ days\\
        \xilinx{Xilinx Vivado Design Suite} & \cmark & \cmark & \cmark & \xilinx{5} & $2$ weeks\\
        \bottomrule
    \end{tabular}
\end{table*}

To address the secure management of hardware \ac{IP}, IEEE standard 1735-2014~\cite{Automation2015} was established to provide aforementioned confidentiality and integrity while ensuring interoperability between \ac{EDA} tools.
The standard is supported by all major tool vendors and is trusted throughout the industry~\cite{Distel2017}.

At CCS'17, Chhotaray et~al.~\cite{Chhotaray2017} presented an in-depth analysis of IEEE~1735.
The authors identified exploitable weaknesses with a focus on cryptographic flaws (e.g., a padding-oracle attack on non-authenticated symmetric encryption). 
They recommend using secure cryptographic algorithms, i.e., authenticated encryption, to protect against these attacks:
\begin{quote}
    \enquote{\textit{From a cryptographic perspective, the solution is simple. Use a provably secure authenticated encryption scheme that supports associated data (AEAD) to encrypt the sensitive IP [...].}}~\cite{Chhotaray2017}
\end{quote}
Our work proves this statement to be insufficient in practice. 
We go considerably further and question the standard's general security (even assuming strong cryptographic primitives) due to its foundational assumptions about the trust model and execution environment:
\begin{quote}
    \enquote{\textit{A decryption tool shall manage its secret or secrets in a private, secure manner. It may harden such secrets directly into its software or use an external persistent storage scheme.}}~\cite[p.~20]{Automation2015}
\end{quote}
Here, the crucial decryption module is embedded within an \ac{EDA} tool that is typically executed in an untrusted, adversary-controlled environment. 
This raises the question of how such a module can be protected so that an adversary cannot exploit the decryption software and disclose valuable \ac{IP}.

\par\smallskip\textbf{Goals and Contributions.}
In this paper, we analyze the security of the IEEE standard~1735-2014 and its implementations.
Our goal is to assess the standard's trust model assumption with respect to the capabilities of real-world adversaries. 
To this end, we review the standard's recommendations on both the handling of cryptographic key material and the toolchain execution environment, with a particular focus on adversaries with static and dynamic software reverse-engineering capabilities.
In the case studies listed in Table~\ref{tab:whitebox:case_study_overview}, we investigate how the standard is realized in market-leading \ac{FPGA} and \ac{ASIC} design and verification tools. 
Based on software reverse engineering, we reveal \textit{all} cryptographic private keys of the analyzed \ac{EDA} tools within weeks, days, or sometimes even just hours.
In particular, we bypass a wide variety of software protection measures, including three different white-box RSA implementations.
Thereby, we provide the first-ever description of public-key white-box schemes used in real-world applications.
We then analyze these schemes and provide multiple attacks, each of which is able to invalidate the standard's security properties of \ac{IP} confidentiality and integrity.
Our main contributions are:
\begin{itemize}
    \item {\bfseries Insecurity of IEEE~1735.}
    We describe flaws in the current IEEE standard 1735-2014 for hardware \ac{IP} protection. 
    Our investigation focuses on the foundational trust model assumptions, i.e., the \textit{untrusted execution environment}, the recommendation of \textit{hard-coded keys}, and the \textit{lack of guidance} on a secure implementation (see Section~\ref{sec:whitebox:standard}).
    We argue that these trust model assumptions cannot be addressed with cryptographic solutions given the current hardware design process and thus render the standard inherently insecure.
    
    \item {\bfseries Case Studies on Market-Leading \ac{EDA} Tools.}
    We analyze implementations of IEEE~1735 in six market-leading \ac{EDA} tools from \xilinx{Xilinx}, \intel{Intel}, \cadence{Cadence}, \siemens{Siemens}, \lattice{Lattice}, and \microsemi{Microsemi}.
    Thereby, we extract \textit{all} cryptographic private keys from these tools leading to a full break of the confidentiality and integrity of the protected \ac{IP}.

    \item {\bfseries Insecure RSA White-Boxes.}
    We analyze three white-box RSA schemes and present distinct cryptanalytical attacks for each of them: a (1)~code-lifting attack that yields a decryption oracle, and a (2)~key extraction attack that recovers the hidden secret key within seconds. 
    Our key extraction attacks exploit knowledge on the custom RSA structure to extract the secret key from the obfuscated decryption functions.
\end{itemize}

\textbf{Responsible Disclosure.}
Following standard responsible disclosure principles, we reported the discovered security vulnerabilities to the affected vendors at least three months ahead of publication, through the vendors' vulnerability disclosure programs or technical support platforms.
All vendors named in this paper have acknowledged our findings, are working towards (or have already implemented) mitigation improvements, and have agreed to publication.
By arrangement with the vendors, we sometimes refrain from providing a detailed elaboration on our reverse engineering process, as this work is not meant to act as a guide to \ac{IP} theft.
Hence, some statements are \textit{deliberately} kept vague.

\section{Background}\label{sec:whitebox:background}
In this section, we provide background on the hardware design flow and the role of \ac{IP} cores.
We also summarize key aspects of IEEE standard 1735-2014 and briefly review white-box cryptography.

\subsection{Hardware Design Flow}\label{subsec:whitebox:design_flow}
Hardware design generally follows a multi-stage process distributed across a global supply chain.
We distinguish between \acp{ASIC} that are fixed in their functionality and \acp{FPGA}, i.e., reconfigurable \acp{IC} that are programmed using a so called \textit{bitstream}.

\begin{figure*}
    \centering
    \includegraphics[width=0.9\textwidth]{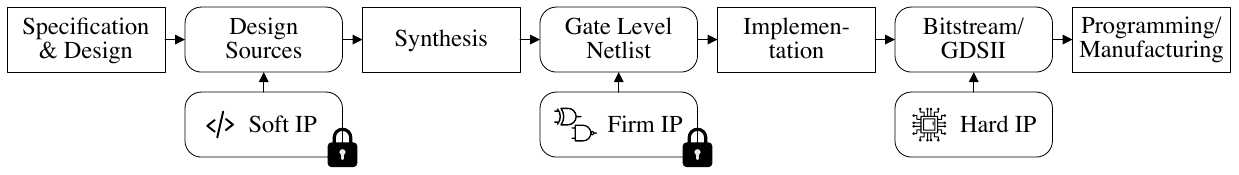}
    \caption{High-level view of the hardware design flow for \acp{FPGA} and \acp{ASIC}. 
    Different types of third-party \ac{IP} can be integrated into the design depending on the design stage.
    IEEE~1735 can be applied to protect both soft and firm \ac{IP}.}
    \label{fig:whitebox:design_flow}
\end{figure*}

\par\smallskip\textbf{Design Flow.}
First, a hardware design is specified and described in an \ac{HDL} such as (System)Verilog or VHDL. 
Next, the design files are analyzed and transformed by a so-called \textit{synthesizer} provided by the \ac{EDA} vendor.
This synthesis step yields a gate-level netlist that is composed of logic gates and their interconnections. 
Next, the netlist is implemented, i.e., technology-mapped, placed, and routed using tools appropriate for the selected technology. 
For \acp{FPGA}, the output of this implementation process is a bitstream file which is then programmed onto the \ac{FPGA}.
\ac{ASIC} implementation yields a layout file (e.g., GDSII) describing the technology-mapped and placed-and-routed netlist then sent to a fab for manufacturing, see Figure~\ref{fig:whitebox:design_flow}.

\par\smallskip\textbf{Types of \ac{IP} Cores.}
Third-party \ac{IP} can be integrated into the hardware design throughout the entire design process.
Three different types of \ac{IP} are distinguished depending on their level of abstraction, cf.~Figure~\ref{fig:whitebox:design_flow}. 
\textit{Soft \ac{IP}} cores are synthesizable descriptions of an \ac{IP} core. 
They offer flexibility to the \ac{IP} user and may be implemented on a wide range of architectures and technologies. 
\textit{Firm \ac{IP}} cores are synthesized gate-level netlists and therefore less flexible than soft \ac{IP}.
The reconstruction of their high-level description is related to the domain of netlist reverse engineering~\cite{Albartus2020,Meade2018a,Subramanyan2014a}.
For \acp{FPGA}, \textit{hard \ac{IP}} cores refer to additional on-device circuitry that can merely be activated and configured by the bitstream, but cannot be retrofitted to the device.
In contrast, hard \acp{IP} in the \ac{ASIC} context comprise a technology-mapped and placed-and-routed design commonly provided as a black-box by manufacturers.

\subsection{IEEE Standard 1735-2014}\label{subsec:whitebox:standard_description}
The IEEE standard 1735-2014~\cite{Automation2015}, also referred to as \enquote{\textit{IEEE Recommended Practice for Encryption and Management of Electronic Design Intellectual Property (IP)}},
refers to a set of guidelines on how to manage and protect electronic design \ac{IP}.
It has been approved and published in late 2015 to promote security and improve interoperability throughout the hardware industry and remains in effect ever since.
Note that, although never explicitly being mentioned in IEEE~1735, its setting is related to classical \ac{DRM}.
A revision of the standard has been announced for 2020 \cite{Krolikoski2020}, but has since been delayed.

\par\smallskip\textbf{Stakeholders.}
The standard defines three main stakeholders.
An \textit{\ac{IP} author} is the creator and legal owner of an \ac{IP} core.
Safeguarding the \ac{IP} is a central objective for the \ac{IP} author in order to protect business interests and prevent potential losses from infringements.
An \textit{\ac{IP} user} acquires rights to an \ac{IP} core through interaction with the \ac{IP} author.
A key interest of the \ac{IP} user is to minimize costs associated with a well-proven \ac{IP} core.
Hence, from a business perspective, there is a tension between an \ac{IP} author and its users.
For this reason, IEEE~1735 explicitly represents an \ac{IP} user as a potential attacker with the goal of infringing on the \ac{IP} author's design.
The standard counters this threat by use of cryptographic protections to safeguard the \ac{IP} core from illegitimate use.
Lastly, the \textit{tool vendor} provides design tools that enable the \ac{IP} user to interact with the \ac{IP} of an author.

\par\smallskip\textbf{Digital Envelope.}
IEEE~1735 dictates the use of a \textit{digital envelope} to protect the \ac{IP} itself as well as the rights granted by the \ac{IP} author.
Hence, the standard specifies a unified format for the protected \ac{IP} to ensure security and compatibility between the stakeholders and their tools. 
It thereby resorts to the use of a mark-up format to customize properties of the \ac{IP} core as well as the applied protection mechanisms. 
IEEE~1735 supports both VHDL and (System)Verilog. 
Hence, its scope is limited to soft and firm \ac{IP} cores, as illustrated by the locks in Figure~\ref{fig:whitebox:design_flow}.

A simplified illustration of the digital envelope is provided in Figure~\ref{fig:whitebox:envelope}.
The envelope is split into two sections, the first of which specifies access rights and additional properties. 
It comprises a number of \textit{right blocks} that specify the rights granted to the \ac{IP} user by its author within each supported \ac{EDA} tool.
For example, such rights can be used to  limit the visible information during simulation or synthesis.
There are two different kinds of rights: those that are common across all supported tools denoted as \textit{common block} and those that are specific to a particular tool denoted as \textit{tool block}. 
The digital envelope comprises one such tool block for every supported \ac{EDA} tool.
Thereby, the standard allows each tool to introduce custom, tool-specific rights.
The second section of the envelope contains the actual \ac{IP} data in encrypted form.

\begin{figure}
    \centering
    \includegraphics[width=0.4\textwidth]{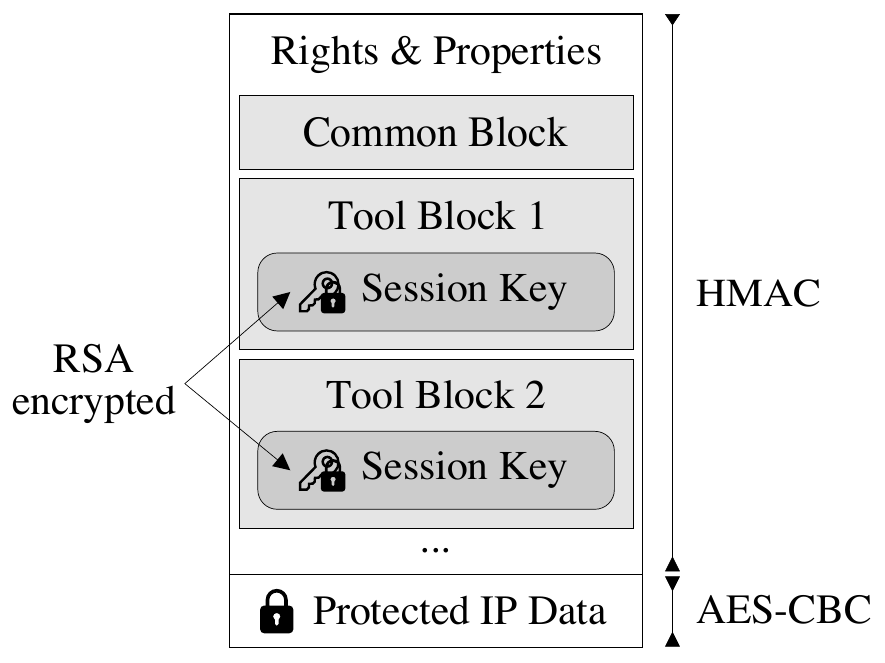}
    \caption{General structure of the digital envelope as specified by IEEE~1735. The envelope comprises the common and tool-specific right blocks, the encrypted session keys, and the protected \ac{IP} data.}
    \label{fig:whitebox:envelope}
\end{figure}

\par\smallskip\textbf{Confidentiality and Integrity.}
As the \ac{IP} data section contains the valuable \ac{IP} sources, upholding its confidentiality is crucial. 
Therefore, the \ac{IP} data is encrypted using AES-CBC and a unique 128- or 256-bit symmetric session key.
This session key is stored within each vendor's tool block after being encrypted using the corresponding tool vendor's RSA public key.
IEEE~1735 mandates the use of RSA keys comprising no less than 2048 bits, though this is not always strictly followed, as is evident from our case studies.
Each tool block contains a unique identifier for the required RSA key pair to support multiple different keys per vendor and facilitate adding new keys over time.

For each tool, the session key is also used to compute an HMAC over the common block as well as its tool block using SHA-256 or SHA-512.
Consequently, the session key needs to be decrypted before the integrity of the RSA key identifier is verified.
The HMAC allows each tool to verify the integrity of the common rights and its own tool-specific rights, but not the \ac{IP} data itself.
Notably, IEEE~1735 does \textit{not} enforce integrity of the encrypted \ac{IP} data and offers \textit{no} authenticity~\cite{Chhotaray2017}.
We emphasize that an exposed session key implies a break of the confidentiality and integrity of the protected IP.

Most \acp{IP} support multiple vendors, hence leakage of a single private key may directly compromise security of the entire scheme.
Therefore, the secure handling of each and every private key is vital.
As not only the private key but also the plaintext \ac{IP} reside in program memory during decryption, recovery of the plaintext \ac{IP} poses an additional threat.

\subsection{RSA}\label{subsec:whitebox:rsa}
In anticipation of our case studies in Sections~\ref{sec:whitebox:case_intel} to~\ref{sec:whitebox:case_xilinx}, we briefly recall RSA encryption and Miller's factorization algorithm~\cite{Miller1976}.

\par\smallskip\textbf{Plain RSA.} 
Let $N = pq$ be an RSA modulus. We denote by $\Z_N$ the ring of integers modulo $N$ with the multiplicative group $\Z_N^{*}$ of order $\varPhi(N) = (p-1)(q-1)$. Let $e \in \Z_{\varPhi(N)}^*$ be a public encryption exponent with corresponding RSA private key $d$ satisfying $ed = 1 \bmod \varPhi(N)$. The RSA encryption function is the map $\Z_N \rightarrow \Z_N, m \mapsto m^e \bmod N$. The RSA decryption function is the map $\Z_N \rightarrow \Z_N, c \mapsto c^d \bmod N$.

\par\smallskip\textbf{Private Key Implies Factorization.}
Given public exponent $e$ and private key $d$, Miller's well-known probabilistic reduction from RSA secret recovery to factoring~\cite{Miller1976} can be applied.
Let $ed-1 = 2^{k}t = 0 \bmod N$ with odd $t$. 
We take a random $g \in \Z_N^*$ until $g^t \not=1 \bmod N$. 
Let us then square $g^t$ until we obtain $1$ as result. 
Let $c < k$ be minimal such that $g^{2^{c+1} t} = 1 \bmod N$. 
Since $1$ has four square roots modulo $N$, we obtain $g^{2^c t} \not= \pm 1 \bmod N$ with probability $\frac 1 2$. In this case
\[
  \gcd( g^{2^c t} \pm 1, N) = \{p,q\}
\]
directly reveals the factorization of $N=pq$.

\subsection{White-Box Cryptography}\label{subsec:whitebox:whitebox_crypto}
We now recall standard notions of white-box cryptography in anticipation of the white-box case studies in Sections~\ref{sec:whitebox:case_microsemi}, \ref{sec:whitebox:case_siemens} and~\ref{sec:whitebox:case_xilinx}.
White-box cryptography and the associated white-box attacker model were introduced in the seminal work of Chow~et~al.~\cite{Chow2002a, Chow2002}. 
The goal of white-box cryptography is to protect a cryptographic secret such as a key within the implementation of a cryptographic algorithm itself.

\par\smallskip\textbf{Attacker Model.}
In traditional \textit{black-box} scenarios, the attacker has access to the cryptographic algorithm, can adaptively choose input plaintexts and/or ciphertexts, and observe the algorithm outputs.
However, the dynamic execution of the algorithm remains hidden from their eyes. 
In contrast, a \textit{white-box} setting grants the attacker full access to the software encompassing the cryptographic algorithm and its execution environment. 
The attacker may thus perform static code analysis, arbitrarily execute the white-box algorithm, examine intermediate values in memory, and dynamically manipulate these values during execution. 
White-box cryptography aims to provide security even in such hostile environments.

\par\smallskip\textbf{Security Goals.}
The two primary security goals of white-box cryptography are formalized as \textit{security against key-extraction} and \textit{security against code-lifting}~\cite{Chow2002a, Chow2002}.
A white-box cryptographic algorithm is secure against key extraction, if and only if an attacker cannot recover the embedded cryptographic secret from the implementation.
However, a white-box scheme that is \textit{only} secure against key extraction might not offer any security at all.
Even without any knowledge of the secret key, an attacker that can extract the white-box algorithm itself can then execute the algorithm and thereby yield an encryption/decryption oracle.
A white-box cryptographic algorithm that provides security against code-lifting attacks cannot be extracted, i.e., \textit{lifted}, from its application while maintaining functionality. 

\par\smallskip\textbf{Public-Key Cryptography.}
Most academic work on the topic deals with symmetric white-box techniques, the only exception -- to the best of our knowledge -- is the work by Barthelemy~\cite{Barthelemy2020} based on a lattice-based scheme.
Furthermore, the CHES 2021 WhibOx contest~\cite{CHES2021} aimed to obfuscate an ECDSA signature scheme to be secure in a white-box setting, however, all submissions have successfully been attacked after at most two days.
While the symmetric white-box setting has extensively been formalized~\cite{AlpirezBock2020,Bock2020,Delerablee2013,Saxena2009}, public-key white-box cryptography lacks any such formalization.
Despite these shortcomings and the limited availability of public research, numerous companies offer public-key white-box solutions~\cite{Digital.ai,Irdeto,Thales}, and hold patents in that area~\cite{Hoogerbrugge2020,Hoogerbrugge2019,Hoogerbrugge2018,Michiels2013,Zhou2009}.
\section{Attacks on IEEE Standard 1735-2014}\label{sec:whitebox:standard}
We now assess the security of IEEE standard 1735-2014.
To this end, we first outline the classic \ac{MATE} attacker model and then identify foundational security flaws in the trust model assumption and the key management of IEEE~1735.
The case studies presented in Sections~\ref{sec:whitebox:case_intel} to~\ref{sec:whitebox:case_xilinx} then demonstrate the severe implications of the discovered flaws in practice.

\subsection{Attacker Model}\label{subsec:whitebox:attacker_model}
\par\smallskip\textbf{Attacker Capabilities.}
Our attacker model is inspired by the well-established \ac{MATE} attack scenario~\cite{Akhunzada2015}.
We consider an attacker with full control over the \ac{EDA} tool execution environment (that contains the secret decryption keys). The attacker is able to thoroughly analyze the \ac{EDA} tool software by performing static and dynamic software analysis. In particular, the attacker is able to interact with the application at runtime and can thus manipulate and extract intermediate values during execution.
Note that this setting directly corresponds to the white-box attacker model outlined in Section~\ref{subsec:whitebox:whitebox_crypto}.

\par\smallskip\textbf{Attacker Goal.}
The attacker's objective is to break the confidentiality and integrity of the protected \ac{IP} to uncover and/or manipulate the plaintext \ac{IP}.
More precisely, the attacker aims to recover the decrypted session key provided within the \ac{EDA} tool rights block of the digital envelope as specified by IEEE~1735 (see Section~\ref{subsec:whitebox:standard_description}). 
Extracting the hard-coded RSA private key from the \ac{EDA} tool fulfills this goal as it allows for decryption of the session key.

\par\smallskip\textbf{Attacker Skill Set}
To contextualize the attack time per case study denoted in Table~\ref{tab:whitebox:case_study_overview}, we briefly summarize our skill set. 
The reverse engineers conducting our case studies are a Master's degree student with eight years of experience in the field of binary analysis and a PhD student working in the area hardware security for five years. 
The attacks on the RSA white-box decryption algorithms were conducted in collaboration with a cryptography expert with 25 years of experience in cryptanalysis.

\subsection{Principle Security Considerations} \label{subsec:whitebox:standard_flaws}
In light of an ongoing revision of the standard~\cite{Krolikoski2020}, we detail flaws in the protocol that break the alleged confidentiality and integrity of the protected \ac{IP}.
Despite not giving explicit strong security guarantees, the standard is trusted throughout the industry:
\begin{quote}
    \enquote{\textit{There are specific IEEE standards for evaluation and comprehension. 
    The effort for cracking this encryption is so high that it is difficult even for large firms.}}~\cite[pp.~139-140]{Distel2017}
\end{quote}

\par\smallskip\textbf{Untrusted Execution Environment.}
Since \ac{EDA} tools commonly run on the workstation(s) of \ac{IP} users, IEEE~1735 operates within a notoriously untrusted execution environment, i.e., an adversary has full control over the execution of said \ac{EDA} tools, see Section~\ref{subsec:whitebox:attacker_model}.
In this hostile setting, the standard attempts to provide confidentiality and (in parts) integrity of the protected \ac{IP} by employing both public-key and symmetric cryptography.
However, even in the presence of impeccable protection measures and analogously to \ac{DRM} use cases, the plaintext \ac{IP} will reside in memory during tool execution and thus can be recovered by an adversary.
To counter this threat, the standard encourages tool vendors to legally prevent reverse engineering of the design tools altogether:
\begin{quote}
    \enquote{\textit{Between tool vendors and their users [...], an agreement for use of the tool should forbid tampering and reverse engineering without being granted explicit permission.}}~\cite[p.~11]{Automation2015}
\end{quote}
However, in reality, an attacker who is willing to commit \ac{IP} theft or insert malicious circuitry into a protected \ac{IP} is unlikely to be stopped by a legal agreement as both attacks are already illegal on their own.

\par\smallskip\textbf{Hard-Coded Keys.}
To our surprise, IEEE~1735 recommends the same RSA private key to be used across all instances of an \ac{EDA} tool.
The standard furthermore proposes to \textit{hard-code} these keys into the design software itself:
\begin{quote}
    \enquote{\textit{One key should be used by a single product or a set of closely related products from one company that shares a common code base.}}~\cite[p.~25]{Automation2015}

    \enquote{\textit{It may harden such secrets directly into its software or use an external persistent storage scheme.}}~\cite[p.~20]{Automation2015}
\end{quote}
As becomes evident by our case studies in Sections~\ref{sec:whitebox:case_intel} to~\ref{sec:whitebox:case_xilinx}, not a single vendor resorts to an \enquote{external storage scheme} (e.g., dongles) as it is regarded to be impractical and not scaleable for most use cases.
An adversary who is able to recover a private key, e.g., by means of reverse engineering, can subsequently decrypt and manipulate each and every protected \ac{IP}.
Hence, this single point of failure constitutes a valuable target for attackers.

\par\smallskip\textbf{Lack of Guidance.}
Although the threat of key extraction is addressed within IEEE~1735, the standard omits a clear discussion on the importance of private key protection and regards the disclosure of \ac{IP} as an implementation issue: 
\begin{quote}
    \enquote{\textit{If a disclosure happens because a tool is hacked during execution by an untrusted IP user, using any combination of binary tampering and debugging techniques, that is an implementation issue. Each tool vendor needs to decide what, if any, anti-debug, anti-tamper, and anti-key-discovery technologies they will use here.}}~\cite[p.~11]{Automation2015}
\end{quote}
The standard mentions that it would be in the tool vendor's best business interest to invest in protecting its private keys, however, no guidance on a secure implementation of the key storage scheme is provided.
Moreover there is a limited incentive to invest into proper software protections for the private keys as only few tool vendors sell and license their own \ac{IP}.
Just as in a \ac{DRM} setting, the security of the entire scheme is founded upon the protection of the private keys, hence the failure to address these threats appropriately has devastating consequences in practice.

\subsection{White-Box Implementations}\label{subsec:whitebox:case_whitebox}
White-box cryptography (see~Section~\ref{subsec:whitebox:whitebox_crypto}), despite not being mentioned in the standard, appears to be a natural fit for the key management scheme introduced by IEEE~1735 due to its \ac{DRM}-like nature. 
Hence, some vendors use white-box implementations to protect their private keys.
Note that we adopt the framing of the discovered algorithms as \textit{white-box} cryptography from the vendors to then evaluate their implementations within the white-box model in Sections~\ref{sec:whitebox:case_microsemi}, \ref{sec:whitebox:case_siemens} and~\ref{sec:whitebox:case_xilinx}.
Thereby, we reveal that these commercial white-box solutions do not fulfill the notion of white-box cryptography and fail to uphold its basic security requirements, i.e., \textit{security against key extraction} and \textit{security against code-lifting}.

\begin{figure}
    \centering
    \includegraphics[width=0.48\textwidth]{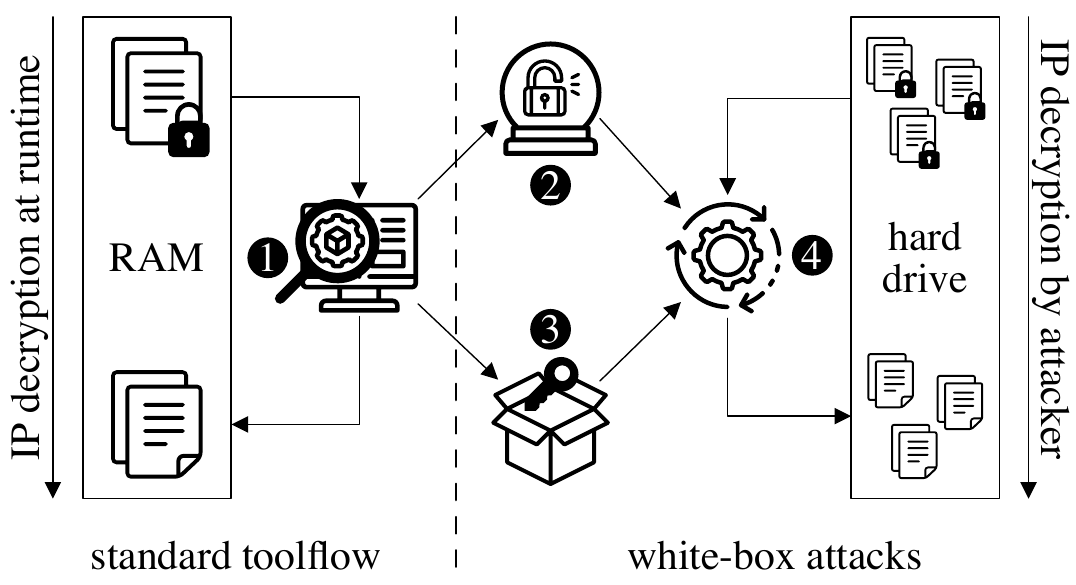}
    \caption{High-level overview of the standard \ac{EDA} tool flow for \ac{IP} decryption as well as code-lifting and key-extraction attacks targeting a white-box implementation.}
    \label{fig:whitebox:attack}
\end{figure}

\par\smallskip\textbf{White-Box Attack Strategies.}
Figure~\ref{fig:whitebox:attack} gives a high-level overview of the standard tool flow for \ac{IP} decryption and our overall attack strategy.
By default, the \ac{EDA} tool loads the protected \ac{IP} into program memory ahead of decryption.
It proceeds to decrypt the symmetric session key using the RSA private key specified within the \ac{IP}'s digital envelope.
The session key is then used for the decryption of the protected \ac{IP} data.
After decryption, the plaintext \ac{IP} resides within program memory awaiting operations such as synthesis or simulation.

Our attack strategy evolves around the two central security properties of white-box cryptography, i.e., \textit{security against code-lifting} and \textit{security against key extraction}.
For every case study involving a white-box, \circled{1}~we analyze the decryption process using static and dynamic software reverse-engineering techniques to locate the cryptographic subroutines.
Having identified the white-box RSA decryption algorithm, \circled{2}~we proceed to lift it from the main application by replicating it in a high-level language.
This replica can be seen as a decryption oracle, as it carries the private key embedded within its implementation.
Given the replicated white-box, we generate an abstract, formalized description of the white-box algorithm. 
Based upon this description, \circled{3}~we develop key-extraction attacks to recover the hidden private key.
From this point on, by keeping a record of either the replicated white-box algorithm or the recovered secret key, \circled{4}~we can decrypt the symmetric session key and consequently uncover the plaintext of any arbitrary \ac{IP} protected by IEEE~1735 and destined for use within the respective tool.
Both attack strategies combined break the two crucial security properties of white-box cryptography.

\section{\intel{Case Study: Intel Quartus Prime}}\label{sec:whitebox:case_intel}
\intel{Since the acquisition of Altera, Intel has grown to be the second largest vendor of \acp{FPGA} in the market. 
Intel Quartus Prime is their main \ac{EDA} tool for \ac{FPGA} design.
Our analysis deals with Intel Quartus Prime in version 21.1 for Windows.
At first, we identify \acp{DLL} of interest by searching for the documented~\cite{IntelIP} RSA key name in program memory at runtime.
Thereby, we uncover two potential target \acp{DLL}, one of which only provides IEEE~1735 encryption capabilities and contains public keys of other vendors.
The other \ac{DLL} is (at least partially) equipped with symbols that, f.i., provide human-readable names to the library functions and thereby support the reverse engineering process.
In order to identify cryptographic subroutines, we automatically scan for unique strings such as \texttt{key\_keyname} using a static analysis tool.
According to IEEE~1735, \texttt{key\_keyname} precedes the actual key identifier within the digital envelope.
By analyzing the string references, we uncover a function called \texttt{assemble\_key}, which prepares the RSA private key for decryption.
Note that the private key is solely protected by XORing of two hard-coded byte-arrays.
Hence, revealing the RSA private key is straight-forward.
}
\section{\cadence{Case Study: Cadence Xcelium}}\label{sec:whitebox:case_cadence}
\cadence{Cadence Design Systems develops \ac{EDA} tools for \ac{ASIC} design that see widespread use throughout the industry.
Cadence Xcelium is a verification tool based on logic simulation, which provides encryption and decryption capabilities in compliance with IEEE~1735.
In this case study, we analyze Xcelium 21.0 on CentOS.
By going through user documentation~\cite{CadenceIP}, we identify two target keys, one of which is an insecure (and by now deprecated) 512-bit RSA key that has been in active use until approximately 2015.
Due to the availability of symbols and the absence of code obfuscation, automated scanning of the executable for the documented RSA key names reveals the target function.
When this function is invoked, one of two other functions is called depending on the selected secret key. 
These functions both return a static memory address pointing to the private key stored in plaintext.
Hence, reverse engineering and recovering the RSA private key is straightforward.
}
\section{\lattice{Case Study: Lattice Radiant}}\label{sec:whitebox:case_lattice}
\lattice{Lattice Semiconductor is one of the four major \ac{FPGA} vendors.
They provide Lattice Radiant to develop hardware designs for their \acp{FPGA}.
By investigating Lattice's implementation of IEEE~1735 within Lattice Radiant 2.2.1 for Windows, we identify the target library by again automatically searching for the documented key name \cite{LatticeIP}. 
Thereby, we discover a function that loads a secret key which is identified by name. 
In order to extract the RSA private key from the decryption context, we follow the references to callers of that function.
We subsequently identify another function that takes as input a key name, the encrypted data, and an empty buffer apparently reserved for the decrypted data. 
The discovered decryption function then calls the RSA decrypt function provided by OpenSSL.
Next, we inject custom code triggering the decryption on the targeted private key and subsequently extract the RSA private key from the OpenSSL decryption at runtime.}
\section{\microsemi{Case Study: Microsemi Libero SoC}}\label{sec:whitebox:case_microsemi}
\microsemi{Microsemi is one of the four major \ac{FPGA} companies and provides solutions for aerospace, defense, communications, and industrial markets.
In this case study, we focus on Libero SoC v2021.1 for Windows, their current \ac{FPGA} design tool.

\subsection{Reverse Engineering}\label{subsec:whitebox:case_microsemi_reversing}
Similar to the first three case studies, automatic scanning for the RSA key name \cite{MicrosemiIP} reveals a valid 2048-bit RSA private key in program memory.
However, since we do not have access to an encrypted \ac{IP} or the corresponding RSA public key, we could not verify its correctness on our own.
After consultation with Microsemi, this key turned out to be a decoy.

We proceed by dynamically analyzing the executable.
To this end, we trigger the RSA decryption routine and observe its operation at runtime by executing Microsemi Libero SoC and importing a specifically crafted \ac{IP}.
We observe that Libero SoC first feeds the encrypted soft \ac{IP} into a synthesizer which in turn yields a synthesized and re-encrypted \ac{IP} at netlist-level.
Next, we turn our attention to the subprocess calls taking place on this re-encrypted netlist, as this is the first time Libero SoC itself interacts with the encrypted files via a dedicated executable.
We observe that this executable is dynamically linked against a \ac{DLL} equipped with code obfuscation which hints at its valuable content.

We again detect instances of the RSA key name within this library, but since static analysis appears to be cumbersome due to obfuscation, we opt for a dynamic approach instead.
However, the binary is equipped with debugger detection mechanisms and terminates standard execution when attaching a debugger. 
We defeat (almost) all anti-debugging measures protecting the \ac{DLL} using ScyllaHide~\cite{scyllahyde}.
Moreover, we discover the use of integrity checks throughout the library, preventing us from inserting software breakpoints, but evade this protection by the use of hardware breakpoints.
We discover the ciphertext, i.e., the RSA-encrypted session key, in memory as it is passed as an input to a function that is called in close proximity to the discovered RSA key name.
Hence, we place a (hardware-based) memory access breakpoint on the ciphertext, which holds execution upon the first access to the monitored memory address.
In our case, the breakpoint is triggered at the entry point to the cryptographic computations.
Thereby, we skip the obfuscated parts of the binary and reveal the (unobfuscated) functions that are responsible for decryption of the session key.

\subsection{Private Key Recovery}\label{subsec:whitebox:case_microsemi_extraction}
The RSA decryption is implemented using Montgomery modular multiplications~\cite{Montgomery1985}, as revealed by static analysis of the unobfuscated cryptographic functions. 
Note that this technique is typically used to speed up subsequent multiplications as required for the modular exponentiation of RSA. 
By identifying the forward and backward transformations into and from Montgomery domain, we discover the main RSA decryption routine between these two function calls.
We observe that the control flow of the decryption differs between executions, utilizing multiple distinct exponentiation algorithms.
By tracing through a single decryption and dumping inputs and outputs for each exponentiation, we reconstruct an entire decryption run.
Our investigation shows an extended version $d_\ell$ of Microsemi's 7680-bit RSA private key $d$ 
to be split into four hard-coded shares, i.e., two 7680-bit numbers $d_1, d_2$ and the two 32-bit values $d_3, d_4$.
Further analysis reveals that $d_\ell$ is computed as 
\begin{equation*}
    d_\ell = d_1 \cdot d_2 + d_3 \cdot 2^{32} + d_4 = d + k \cdot \varPhi(N),\: k \in \mathbb{N}.
\end{equation*}
Note that Microsemi refers to this implementation as a white-box, hence we have indicated this within Table~\ref{tab:whitebox:case_study_overview} using braces.

For recovery of the original 7680-bit private RSA key $d = d_\ell \bmod \varPhi(N)$, we factor $N=pq$ to compute $\varPhi(N) = (p-1)(q-1)$. 
To this end, we leverage the well-known factorization approach outlined in Section \ref{subsec:whitebox:rsa}, which also works for $d_\ell$.}
\section{\siemens{Case Study: Siemens ModelSim}}\label{sec:whitebox:case_siemens}
\siemens{Ever since incorporating Mentor Graphics in 2016, ModelSim is Siemens' tool for \ac{RTL} as well as netlist-level simulation and comes with support for IEEE~1735~\cite{SiemensIP}.
Siemens ModelSim is commonly bundled with other \ac{EDA} tools and \ac{IP} encrypted for these tools often includes an additional tool block for ModelSim, see Section \ref{subsec:whitebox:standard_description}. 
Hence, extracting the private RSA keys from Siemens ModelSim impacts multiple vendors at once.

\subsection{Reverse Engineering}\label{subsec:whitebox:case_siemens_reversing}
We target ModelSim Pro version 2020.4 for Windows.
In order to trigger and inspect the decryption at runtime, we load a protected \ac{IP} into the \ac{EDA} tool similar to previous case studies.
As we could not identify cryptographic functions within the main ModelSim executable using static analysis, we attach a debugger and track all sub-process calls.
Thereby, we identify a sub-process that handles decryption of the protected IP, as indicated by its console output. 
Using static analysis, we verify the use of a well-known cryptographic library within this sub-process, but are unable to identify all cryptographic functions right away.
However, the binary is still equipped with assert statements, which include the file path and line number of the original asserts within the source code.
Thus, these asserts help us to identify the missing functions by matching file names and line numbers with the open-source cryptographic library implementation.

\subsection{Code-Lifting Attack}\label{subsec:whitebox:case_siemens_lifting}
Through dynamic analysis, we discover that the standard RSA decryption is not called during execution.
Instead, a custom decryption routine is implemented using the \ac{API} provided by the cryptographic library. 
Another assert statement discovered within the binary hints at an RSA white-box implementation using the \ac{CRT}.

In order to recover the session key, we first opt for a code-lifting attack that reconstructs the white-box decryption to be used as an oracle, see Section~\ref{subsec:whitebox:whitebox_crypto}.
To this end, we reconstruct the white-box algorithm in a high-level language.
We verify correct functionality by decrypting an arbitrary encrypted \ac{IP} using our implementation.
During analysis we discover support for three distinct RSA key pairs, two 1024-bit and one 2048-bit pair.
The same decryption routine is called for all three private RSA keys.
Given the white-box data associated with each key, we can then decrypt protected IP even without knowledge of the keys.

This reconstruction does not only break the essential white-box security notion of \textit{security against code-lifting attacks}, but also serves as the basis for the understanding of the underlying white-box algorithm as well as the extraction of the secret keys in the next section.

\subsection{Private Key Extraction}\label{subsec:whitebox:case_siemens_obfuscation}

\par\smallskip\textbf{\acs{CRT}-RSA.}
Let us recall standard \acs{CRT}-RSA decryption as shown in Algorithm~\ref{alg:whitebox:rsa_crt_exp}. 
Let $d_p = d \bmod p-1$ and $d_q = d \bmod q-1$. 
We compute $c^d \bmod N$ via $c^d = c^{d_p} \bmod p$ and $c^d = c^{d_q} \bmod q$, and eventually combine both results using the \ac{CRT}.

\begin{algorithm}[H]
    \caption{\textsc{CRT-Exp}($c, d , d_p, d_q, p, q, \gamma, N$)}
    \label{alg:whitebox:rsa_crt_exp}
    \begin{algorithmic}[1]
        \State{$m_p = c^{d_p} \bmod p.$}
        \State{$m_q = c^{d_q} \bmod q.$}
        \State{$m = (m_p - m_q)\gamma + m_q \bmod N$}\Comment{apply CRT}
        \State\Return{$m$}
    \end{algorithmic}
\end{algorithm}
\noindent The CRT-parameter $\gamma$ satisfies 
\begin{align}
\label{eq:whitebox:crtparam}
   \gamma = 
   \begin{cases}
   1 \bmod p, \\
   0 \bmod q.
   \end{cases}
\end{align}
Thus, in Step 3 of Algorithm~\ref{alg:whitebox:rsa_crt_exp} we compute $m \in \Z_N$ satisfying 
\begin{align*}
   m = 
   \begin{cases}
   m_p = c^{d_p} = c^d \bmod p, \\
   m_q = c^{d_q} = c^d \bmod q.
   \end{cases}
\end{align*}

Computing \textsc{CRT-Exp} is four times faster than computing $c^d \bmod N$ directly, hence it is commonly applied in practice to speed up RSA exponentiations.
However, \textsc{CRT-Exp} is also much harder to obfuscate, since \textit{all five} parameters $d_p, d_q, p, q, \gamma$ directly reveal the factorization of $N$. 
This is obvious for $p,q$. 
For $d_p$ we have $c^{d_p}-c = 0 \mod p$ for any $c$, and therefore $\gcd(N, c^{d_p}-c) = p$, analogous for $d_q$. 
And last, by Equation~(\ref{eq:whitebox:crtparam}) we have $\gcd(N, \gamma-1) = p$ and $\gcd(N, \gamma) = q$.

\medskip
\par\smallskip\textbf{Obfuscated \acs{CRT}-RSA.}
We now outline how Siemens obfuscates the critical values $d_p, d_q, p, q, \gamma$, where the first four are $512$-bit each.

\begin{description}
\item[$d_p, d_q$:] Additively split $d_p = d_{p,1} + d_{p,2}$ and $d_q = d_{q,1} + d_{q,2}$. 
\item [$p,q$:] Use moduli $\tilde p = kp$, $\tilde q = \ell q$ with $8$-bit integers $k$ and $\ell$. 
Notice that if we compute $m_p = c^{d_p} \bmod \tilde p$, then $m_p$ also holds modulo $p$. 
Thus, the correctness of Algorithm~\ref{alg:whitebox:rsa_crt_exp} is maintained. 
However, $\tilde p$ and $\tilde q$ still reveal the factorization by computing $p=\gcd(\tilde p, N)$ or $q=\gcd(\tilde q, N)$. 
Therefore, Siemens ModelSim additively splits $\tilde p = p_1 - p_2$ and $\tilde q = q_1 - q_2$, where $p_2, q_2$ are 256-bit numbers. 
Siemens ModelSim's sophisticated obfuscation technique \textsc{Obf-Mod} now performs a reduction modulo $\tilde p$ using only the split $p_1, p_2$, see Algorithm~\ref{alg:whitebox:reduction}. 
\item[$\gamma$:] The parameter is split multiplicatively and additively, thereby yielding $\gamma = \gamma_1(\gamma_2- \gamma_3)$. 
\end{description}
Altogether, this results in the obfuscated \acs{CRT}-RSA exponentiation shown in Algorithm~\ref{alg:whitebox:obf_rsa_crt_exp}.

\begin{algorithm}[H]
    \caption{\textsc{Obf-CRT-Exp}($c$, $d_{p,1}$, $d_{p,2}$, $d_{q,1}$, $d_{q,2}$, $p_1$, $p_2$, $\gamma_1$, $\gamma_2$, $\gamma_3$, $N$)}
    \label{alg:whitebox:obf_rsa_crt_exp}
    \begin{algorithmic}[1]
        \State{$m_p = \textsc{Obf-Mod}(c^{d_{p,1}} \cdot c^{d_{p,2}}, p_1, p_2)$}
        \State{$m_q = \textsc{Obf-Mod}(c^{d_{q,1}} \cdot c^{d_{q,2}}, q_1, q_2)$}
        \State{$m = (m_p - m_q)\gamma_1(\gamma_2 + \gamma_3) + m_q \bmod N$}\Comment{apply CRT}
        \State\Return{$m$}
    \end{algorithmic}
\end{algorithm}

Let us now have a closer look at \textsc{Obf-Mod} in Algorithm~\ref{alg:whitebox:reduction} that performs the obfuscated modular reduction of some parameter $a$ modulo $\tilde p =  p_1 - p_2$ using only $p_1, p_2$. 

\begin{algorithm}[H]
    \caption{\textsc{Obf-Mod}($a, p_1, p_2$)}
    \label{alg:whitebox:reduction}
    \begin{algorithmic}[1]
        \State{Set $q = \lfloor \frac a {p_1} \rfloor$.}
        \State{Set $r = a \bmod p_1$.}
        \State\Return{$q \cdot p_2 + r \bmod N$}
    \end{algorithmic}
\end{algorithm}
\noindent Notice that \textsc{Obf-Mod} maps $a$ to the value
\begin{align*}
    q \cdot p_2 + r & = \lfloor \frac {a}{p_1} \rfloor  \cdot p_2 + (a \bmod p_1) \\
    & = \lfloor \frac a {p_1} \rfloor  \cdot p_2  + a - \lfloor \frac a {p_1}\rfloor \cdot p_1 \\
    & = a - \lfloor \frac a {p_1} \rfloor  \cdot \left(p_1-p_2\right) \\
    & = a - \lfloor \frac a {p_1}\rfloor  \cdot \tilde p.
\end{align*}
Thus, \textsc{Obf-Mod} indeed reduces $a$ by a multiple of $\tilde p$. 
However, notice that \textsc{Obf-Mod} \textit{does not} realize the standard Euclidean reduction $a - \lfloor \frac a {\tilde p}\rfloor  \cdot \tilde p$ with a remainder in $[0,\tilde p)$. Nonetheless, since $p_2 \ll p_1$ we have $\tilde p \approx p_1$, and therefore $a$ gets sufficiently reduced.

\par\smallskip\textbf{Private Key Recovery.} 
We extract $\gamma = \gamma_1(\gamma_2 + \gamma_3)$ and compute the factorization of $N$ as $\gcd(\gamma-1, N) = p$ and $\gcd(\gamma, N)=q$.
For verification we check whether for $\tilde p = kp = p_1-p_2$ and $\tilde q = \ell q = q_1-q_2$ it holds that $\gcd(N, \tilde p) = p$ and $\gcd(N, \tilde q) = q$. 
Eventually, we compute $d_p = d_{p,1} + d_{p,2}$ and $d_q = d_{q,1} + d_{q,2}$ to verify that for random $c \in \Z_N$ we have $\gcd(N, c^{d_p}-c)=p$ and $\gcd(N, c^{d_q}-c)=q$.
Recovering $d$ now is as simple as $d = e^{-1} \bmod \varPhi(N)$.

}
\section{\xilinx{Case Study: Xilinx Vivado Design Suite}}\label{sec:whitebox:case_xilinx}
\xilinx{Xilinx is the world's largest manufacturer of \acp{FPGA}.
\textit{Xilinx Vivado Design Suite} is their current \ac{EDA} tool for \ac{FPGA} development.
We target \textit{Xilinx Vivado} version 2020.2 for Windows, but have verified applicability to 2019.2 and~2020.1.

\par\smallskip\textbf{Remark.}
Xilinx makes use of white-box cryptography to protect (some of) its RSA private keys.
After concluding our work on Xilinx Vivado in August 2020, we discovered another white-box attack~\cite{widevine} on Google's Widevine \ac{DRM} solution published on October 31, 2020.
We verified that Widevine used the same white-box scheme as Xilinx, hence indicating that the white-box is supplied by a third-party vendor.

\subsection{Reverse Engineering}\label{subsec:whitebox:case_xilinx_reversing}
As in previous case studies, we begin by exploring the Xilinx Vivado executable using static analysis.
We locate implementations of cryptographic primitives by scanning the program memory for cryptographic constants during decryption.
Thereby, we determine that Xilinx Vivado uses a well known cryptographic library to handle its private key.
Moreover, Xilinx secures Vivado using control-flow obfuscation, debugger detection, and exception-based anti-debugging techniques among others, see~\cite{Chen2016,Chen2008,Collberg1997} for a comprehensive overview on software protection techniques.
However, the discovered cryptographic functions do not exhibit any such software protections.

The targeted Xilinx Vivado installation supports five different 2048-bit RSA key pairs \cite{XilinxIP}. 
Our analysis subsequently reveals one of the five RSA private keys using dynamic analysis.
We verify correctness by decryption of an openly-available \ac{IP} encrypted using the corresponding public key.

\subsection{Code-Lifting Attack}\label{subsec:whitebox:case_xilinx_lifting}
Using dynamic analysis, we observe that the discovered RSA decryption routine is no longer triggered for the four other RSA private keys.
In order to identify the decryption routine used for these keys, we statically follow the call graph of the discovered decryption, aiming to locate the dispatcher calling the respective decryption functions for the different keys.
However, large parts of the functions within the call-graph (excluding cryptographic primitives) are protected by software obfuscation.
By getting around these protective measures, we reveal a reference to another presumably cryptographic function.
Observing the inputs and outputs of the newly discovered function using debugging reveals that it indeed performs an RSA decryption and is only invoked for the remaining four RSA private keys.
Due to the non-standard RSA decryption implementation that is heavily based on custom precomputed tables, we conclude that the function implements an RSA white-box.
By triggering the RSA decryption on different private keys, we observe that the same code is executed for all four available white-box keys, but different data structures are accessed depending on the selected key.

To obtain a decryption oracle that allows to recover the RSA-encrypted session keys, we perform a code-lifting attack by replicating the white-box RSA decryption in Python.
The replication of the white-box decryption breaks the white-box property of \textit{security against code-lifting attacks}.

\subsection{White-Box Simplification}\label{subsec:whitebox:case_xilinx_simplification}
In anticipation of the private key extraction attack in the next section, we briefly describe high-level implementation concepts and our simplifications thereof.
Note that manipulations to the implementation of the RSA white-box are in line with our attacker model defined in Section~\ref{sec:whitebox:standard}.

\par\smallskip\textbf{Montgomery Multiplications.} 
First, we identify the use of Montgomery multiplications~\cite{Montgomery1985} by identifying Montgomery-specific constants such as $R = 2^{|d|} \bmod N$ for private key bit-length $|d|$ and modulus $N$ within the extracted data structures.
To simplify the algorithmic description of the white-box implementation, we replace all Montgomery transformations and multiplications with isomorphic standard modular operations.

\par\smallskip\textbf{Ciphertext Dependency.} 
The white-box algorithm utilizes pre-computed values, some of which are selected based on the input ciphertext.
This ciphertext-dependency is introduced by first evaluating a hash function on the ciphertext and then applying multiplicative masking to the ciphertext based on the hash function output.
We observe that these values do not impact the decryption as verified by fixing the hash function output to zero and observing the output over multiple decryption runs.
Therefore, we conclude that we can safely remove this ciphertext-dependency.

\subsection{White-Box Description}\label{subsec:whitebox:case_xilinx_whitebox}
This section gives a formal description of the white-box RSA decryption routine we uncovered from within Xilinx Vivado.
Based upon this description, we subsequently present two sophisticated attacks that instantly lead to a full recovery of the protected RSA secret key in Section~\ref{subsec:whitebox:case_xilinx_extraction}.

\par\smallskip\textbf{Notation.}
We denote vectors as bold-face, lowercase letters such as \textbf{v} and matrices as bold-face, uppercase letters like~\textbf{A}. 
We address vector and matrix entries as $\vec v[i]$ and $\vec A[i][j]$ respectively. 

\par\smallskip\textbf{Ordinary RSA Decryption.}
Let us first consider an ordinary RSA decryption. 
Our goal is to recover the plaintext $m$ by computing the exponentiation $c^d = m \bmod N$. 
We introduce a window technique whose usual purpose is to provide a time-memory trade-off for fast exponentiation, but in our case additionally lies at the heart of Xilinx Vivado's white-box implementation.
Consider a secret decryption key $d$ of length $n$ that is split into $k = \lceil \frac{n}{5} \rceil$ chunks of five bits each, i.e., 
\[ 
d = d_0 + d_1\cdot 2^5 + d_2 \cdot 2^{10} + \ldots + d_k \cdot 2^{5k} \textrm{ with } d_i \in \mathbb{Z}_{32}.
\]
We thus write $\mathbf{d} =  (d_0, d_1, \ldots, d_k) \in \mathbb{Z}_{32}^{k+1}$ for the vector containing these secret key chunks. Then it follows that
\begin{align}
 m & = c^d = c^{d_0 + d_1\cdot 32 + d_2 \cdot 32^{2} + \ldots + d_k \cdot 32^{k}} = \prod_{i=0}^k \left(c^{d_i}\right)^{32^i} \nonumber \\
 & = \left(\left( (c^{d_k})^{32} \cdot c^{d_{k-1}} \right)^{32} \cdot \ldots \cdot c^{d_1} \right)^{32} \cdot c^{d_0} \bmod N.
\label{eq:rsa}
\end{align}
Note that we can precompute the ciphertext-dependent vector $ \mathbf{c} = (c^0, c^1, \ldots, c^{31}) \in \mathbb{Z}_N^{32}$ once at the beginning and use it as a lookup table throughout the decryption process. Thus, the computation of Equation~(\ref{eq:rsa}) is realized by procedure \textsc{EXP}($\textbf{c}, \textbf{d} $) in Algorithm~\ref{alg:whitebox:rsa_exp} comprising only $k$ multiplications and exponentiations in $\mathbb{Z}_N$.
\begin{algorithm}[H]
    \caption{\textsc{Exp}($\textbf{c}, \textbf{d} $)}
    \label{alg:whitebox:rsa_exp}
    \begin{algorithmic}[1]
        \State{$m = \mathbf{c}[\mathbf{d}[k]] = c^{d_k}$}
        \For {$i=k-1$ \textbf{downto} $0$}
            \State{$m = m^{32} \cdot \mathbf{c}[\mathbf{d}[i]] = m^{32} \cdot c^{d_i} \bmod N$}
        \EndFor
        \State\Return{$m$}
    \end{algorithmic}
\end{algorithm}

\par\smallskip\textbf{Hiding $\mathbf{d}$.} 
Notice that in its current form, Algorithm~\ref{alg:whitebox:rsa_exp} directly operates on (and thereby reveals) $\mathbf{d} = (d_0, \ldots, d_k)\in \mathbb{Z}_{32}^{k+1}$ as input.
Hence, an attacker could straightforward recover the secret key $d = \sum_{i=0}^k d_i \cdot 32^i$ from memory. 

The core idea behind the discovered white-box approach is to hide all original $5$-bit values $d_i \in \mathbb{Z}_{32}$ of $\textbf{d}$ via some secret \textit{permutation} $\pi: \mathbf{Z}_{32} \rightarrow \mathbf{Z}_{32}$. 
The bijection $\pi$ is hard-coded within the implementation in an obfuscated manner.
Each of the four available secret keys using the white-box utilizes a different $\pi$.

Notice that $\pi$ as well as the inverse permutation $\pi^{-1}$ can efficiently be represented by providing a lookup table with values $\pi(0), \ldots, \pi(31)$ requiring only $32 \cdot 5=160$ bits. 

As a result, instead of $\mathbf{d}$ itself, only the obfuscated key vector $\hat{\mathbf{d}} = (\pi^{-1}(d_0), \ldots, \pi^{-1}(d_k))$ is revealed during execution of the white-box algorithm. 
Since the computation $\pi(\hat{\mathbf{d}}) = \mathbf{d}$ immediately yields the original secret key, an attacker may find considerable value in trying to recover the secret permutation $\pi$. 
However, the number of permutations on $\mathbf{Z}_{32}$ is given as $32! > 2^{117}$. 
Therefore, trying to guess $\pi$ using a brute-force approach is infeasible. 
This also provides reasoning for the choice of a chunk size of $k=5$ bits, since it is the smallest $k$ that provides sufficient security against brute-force attacks.

Now, let us precompute the vector $\bar{\mathbf{c}} = (c^{\pi(0)}, \ldots, c^{\pi(31)})$. When subsequently executing \textsc{Exp}($\bar{ \textbf{c}}, \hat{\textbf{d}}$) on the permuted inputs $(\bar{ \textbf{c}}, \hat{\textbf{d}})$ instead of $(\textbf{c}, \textbf{d})$, the intermediate values
\[
  \bar{ \textbf{c}}[\hat{\textbf{d}}[i]] = \bar{ \textbf{c}}[\pi^{-1}(d_i)] = c^{\pi(\pi^{-1}(d_i))} = c^{d_i}.
\]
are computed.
Thus, \textsc{Exp}($\bar{ \textbf{c}}, \hat{\textbf{d}}$) outputs the same plaintext $m=c^d$ as \textsc{Exp}(${ \textbf{c}}, {\textbf{d}}$) does on the original inputs.

\par\smallskip\textbf{Randomization of $\mathbb{\bar{ \textbf{c}}}$.} 
Notice that an attacker knows the ciphertext $c$ and thus can compute $c^{0}, \ldots, c^{31}$ himself.
As a result, $\mathbb{\bar{ \textbf{c}}} = (c^{\pi(0)}, \ldots, c^{\pi(31)})$ would immediately reveal the secret permutation $\pi$. 
In order to prevent this disclosure of $\pi$, 
a randomized and obfuscated precomputation vector ${\hat{ \textbf{c}}}$ has to be defined. 
To this end set
\begin{align}
  \mathbf{r} & = (r_0, \ldots, r_{31}) \textrm{ for some } r_i \in_R \mathbb{Z}_N^* \textrm{ and} \nonumber \\ 
  {\hat{ \textbf{c}}} & = \mathbf{r} \cdot \mathbf{\bar{\textbf{c}}} = (r_{\pi(0)} \cdot c^{\pi(0)}, \ldots, r_{\pi(31)} \cdot c^{\pi(31)}). \label{eq:whitebox:Obfuscated_c}
\end{align}
Since the $r_i$ are chosen uniformly at random, $\mathbb{\hat{ \textbf{c}}}$ does not reveal any information about $\pi$. 
Looking at the output of an execution of \textsc{Exp}($\hat{ \textbf{c}}, \hat{\textbf{d}}$), we notice that the intermediate values are now of the form 
\[
  \hat{ \textbf{c}}[\hat{\textbf{d}}[i]] = \hat{ \textbf{c}}[\pi^{-1}(d_i)] = r_{\pi(\pi^{-1}(d_i))} \cdot c^{\pi(\pi^{-1}(d_i))} = r_{d_i} \cdot c^{d_i}.
\]
Given Equation~(\ref{eq:rsa}), it is not hard to see that \textsc{Exp}($\hat{ \textbf{c}}, \hat{\textbf{d}}$) therefore produces the output 
\begin{align*}
 &  \left( \left( (r_{d_k} \cdot c^{d_k})^{32} \cdot r_{d_{k-1}}c^{d_{k-1}} \right)^{32} \cdot \ldots \cdot r_{d_1}c^{d_1} \right)^{32} \cdot r_{d_0}c^{d_0} \\
 & = c^d \cdot \prod_{i=0}^k r_{d_i}^{32^i} = m \cdot  \prod_{i=0}^k r_{d_i}^{32^i} \bmod N.
\end{align*}
Now define the constant term $r = \prod_{i=0}^k r_{d_i}^{-32^i}  \bmod N$. 
In order to reconstruct the desired RSA plaintext $m$ one has to multiply the output of \textsc{Exp}($\hat{ \textbf{c}}, \hat{\textbf{d}}$) by $r$. 
Subsequently we obtain 
\begin{equation}
\label{eq:whitebox:dec_correctness}
  c^d = \textsc{Exp}({ \textbf{c}}, {\textbf{d}}) = r \cdot \textsc{Exp}(\hat{ \textbf{c}}, \hat{\textbf{d}}).
\end{equation}

During the precomputation step an attacker can not directly compute ${\hat{ \textbf{c}}} = (r_{\pi(0)} \cdot c^{\pi(0)}, \ldots, r_{\pi(31)} \cdot c^{\pi(31)})$, 
since the implementation hides the values of $r_i$ and $\pi(i)$ for $i \in \{0, \ldots, 31\}$ from the attacker. 
To this end, the discovered cryptographic algorithm presents a clever approach towards realizing a white-box RSA decryption.

\par\smallskip\textbf{Realization of $\textbf{r}$ and $\pi$.}
As revealed by the reverse engineering efforts described in Section~\ref{subsec:whitebox:case_xilinx_lifting}, 
Xilinx Vivado executes the white-box RSA precomputation shown in Algorithm~\ref{alg:whitebox:whitebox_pre} on any given ciphertext $c \in \mathbb{Z}_N$.
Algorithm~\ref{alg:whitebox:whitebox_pre} computes a vector $\textbf{s} \in \mathbb{Z}_N^{33}$ comprising powers $s_i = (\alpha c + \beta)^i$ of the linearly transformed ciphertext $c$ 
for all $0 \leq i \leq 32$ using some hard-coded constants $\alpha, \beta \in \mathbb{Z}_N$.
As we will show in the following, the fixed matrix $\vec A \in \mathbb{Z}_N^{32 \times 33}$ utilized in Algorithm~\ref{alg:whitebox:whitebox_pre} realizes 
(and hides) the secret permutation $\pi$ on $\mathbb{Z}_{32}$ and the randomization vector $\mathbf{r} = (r_0, \ldots, r_{31})$ from Equation~(\ref{eq:whitebox:Obfuscated_c}).

\begin{algorithm}
    \caption{\textsc{Obfusc-Precomp}($c$)}
    \label{alg:whitebox:whitebox_pre}
    \begin{algorithmic}[1]
        \State Use hard-coded $ \mathbf{A} \in \mathbb{Z}_N^{32 \times 33}, \textbf{t'} = (t_0, \ldots, t_{31}) \in \Z_N^{32}, \alpha, \beta \in \mathbb{Z}_N$.
        \State{$\textbf{s} =(s_0, \ldots, s_{32})$ with $s_i = (\alpha c + \beta)^i \bmod N$ }
        \State{$\textbf{t} = -c^{32} \cdot \vec t'$ }
        \State\Return{$\mathbf{As} + \mathbf{t} = {\hat{ \textbf{c}}} = (r_{\pi(0)} \cdot c^{\pi(0)}, \ldots, r_{\pi(31)} \cdot c^{\pi(31)})$}
    \end{algorithmic}
\end{algorithm}
Given just the description of Algorithm~\ref{alg:whitebox:whitebox_pre}, it remains unclear how the output $\mathbf{As} + \mathbf{t}$ realizes an obfuscated vector $\hat{ \textbf{c}}$. 
Using the \textit{binomial theorem}, we can write $s_i$ as
\begin{align*}
    s_i & =  (\alpha c + \beta)^i = \sum_{j=0}^i \binom{i}{j}   \left( \alpha c \right)^{j} \beta^{i-j}\quad \forall i : 0 \leq i\leq 32.
\end{align*}
Thus, within the sum making up $s_i$ the coefficient of $c^{j}$ with $j \leq i$ is given as $\binom{i}{j}  \alpha^{j} \beta^{i-j}$. 
Let us subsequently define the following lower-triangular matrix $\vec M \in \mathbb{Z}_N^{33 \times 33}$ with entries 
\[\vec{M}[i][j] = \binom{i}{j}\alpha^j \beta^{i-j} \quad \forall i,j: 0 \leq j \leq i \leq 32,
\]
and $\vec{M}[i][j] = 0$ otherwise: 
\[
\small{
  \mathbf{M} = 
  \left(
  \begin{array}{cccccc}
  1 &  &  &  &  & \\
  \beta & \alpha &  &  &  & \\
  \beta^2  & 2\alpha\beta & \alpha^2 &  &  &   \\
  \beta^3 & 3 \alpha\beta^2 & 3 \alpha^2 \beta & \alpha^3 & &  \\
   \vdots & \vdots & \vdots & \vdots &   \ddots \\
  \beta^{32}& 32 \alpha \beta^{31} & \binom{32}{2} \alpha^2 \beta^{30} & \binom{32}{3} \alpha^3 \beta^{29}  & \ldots &  \alpha^{32} 
  \end{array}
  \right).
}
\]
Note that we omit the zeros from the depiction of $\textbf{M}$ for ease of presentation. 

Let us define the vector $\vec c = (1, c, c^2, \ldots, c^{32})$ made up of powers of the ciphertext $c$. 
Then by definition it follows that
\[
   \vec M \cdot \vec c = \vec s.
\]

The following theorem shows that Algorithm~\ref{alg:whitebox:whitebox_pre} indeed computes the obfuscated vector $\hat{ \textbf{c}}$ if and only if $\mathbf{AM}$ directly reveals the secret permutation $\pi$ as well as the randomization vector $\mathbf{r}$.

\begin{theorem}
\label{theo:whitebox:main}
  Let $\vec e_i \in \{0,1\}^{32}$ be the $i^{\textrm{th}}$ unit vector. We have $\mathbf{As} + \mathbf{t} = {\hat{ \textbf{c}}} = (r_{\pi(0)} \cdot c^{\pi(0)}, \ldots, r_{\pi(31)}  \cdot c^{\pi(31)})$ if and only if 
  \[
      \mathbf{AM} = 
  \left(
  \begin{array}{cl}
  r_{\pi(0)} \vec{e_{\pi(0)}} & t_0 \\
  r_{\pi(1)}  \vec{e_{\pi(1)}} & t_1 \\  
  \vdots & \vdots  \\
  r_{\pi(31)}  \vec{e_{\pi(31)}} & t_{31} \\  
  \end{array}
  \right). 
  \]
In other words, for $0 \leq i < 32$ the $i^{\textrm{th}}$ row vector of $\mathbf{AM}$ is entirely made up of $0$'s, 
except for the last and the $\pi(i)^{th}$ position, at which we get entries $t_i$ and $r_{\pi(i)} $ respectively.
\end{theorem}
\begin{proof}
Let $\vec{a}_i$ be the $i^{\textrm{th}}$ row of $\vec A$. We then have to show that if for all $0 \leq i \leq 33$ we have $\vec{a}_i \cdot \vec s + \vec t[i] = r_{\pi(i)} c^{\pi(i)}$ it follows that
\begin{equation}
\label{eq:whitebox:identity}
  \vec a _i \cdot \vec{M} = (r_{\pi(i)} \vec{e_{\pi(i)}},\:t_i),
\end{equation} 
and vice versa. 
Hence, assume that $\vec a_i \cdot \vec s + \vec t[i] = r_{\pi(i)} c^{\pi(i)}$. Since we have $\vec s =  \vec M \cdot \vec c$ and $\vec t[i] = -t_i c^{32}$, we obtain
\begin{equation}
\label{eq:whitebox:identity2}
  \vec a_i \cdot \vec M \cdot \vec c = r_{\pi(i)} c^{\pi(i)} + t_i c^{32} = (r_{\pi(i)}  \vec{e_{\pi(i)}},\: t_i) \cdot \vec c.
\end{equation}
As the identity from Equation~(\ref{eq:whitebox:identity2}) has to hold independent of the ciphertext $c$ and thus for all $\vec c$, this already implies the desired identity from Equation~(\ref{eq:whitebox:identity}).

For the opposite direction, assume that Equation~(\ref{eq:whitebox:identity}) holds. This immediately implies Equation~(\ref{eq:whitebox:identity2}), from which we conclude that $\vec{a}_i \cdot \vec s + \vec t[i] = r_{\pi(i)} c^{\pi(i)}$.
\end{proof}

\subsection{Private Key Extraction}\label{subsec:whitebox:case_xilinx_extraction}
In the following, we propose two attack strategies to recover the secret permutation $\pi$ from the white-box cryptographic algorithm.
Hence, we show that the recovered white-box algorithm does not hold up to the notion of \textit{security against key extraction} as outlined in Section~\ref{subsec:whitebox:whitebox_crypto}. 

\par\smallskip\textbf{Chosen Ciphertext Attack.}
Extracting the key via a chosen ciphertext attack is straightforward once the ciphertext-dependency has been removed as outlined towards the end of Section \ref{subsec:whitebox:case_xilinx_lifting}.
Given the modified white-box algorithm, we first run \textsc{Obfusc-Precomp}($1$) from Algorithm~\ref{alg:whitebox:whitebox_pre} using ciphertext $c=1$ as input. 
By the correctness property of Algorithm~\ref{alg:whitebox:whitebox_pre}, we obtain the output 
\[
\pi(\mathbf{r}) =(r_{\pi(0)}, \ldots, r_{\pi(31)}).
\]
Next, we run \textsc{Obfusc-Precomp}($2$), resulting in output 
\[
(r_{\pi(0)}  2^{\pi(0)}, \ldots, r_{\pi(31)}  2^{\pi(31)}).
\]
After dividing each element of the second output by the respective $r_i$, we receive vector $\vec v = (2^{\pi(0)}, \ldots, 2^{\pi(31)}) \in \Z_N^{32}$ comprising powers of $2$. 
Since $\pi(i) \leq 31$ and $2^{31} < N$, no entry is reduced modulo $N$ such that $\vec v \in \Z^{32}$.
Therefore, computing the $\log_2$ of $\vec v$ element-wise yields the values of $\pi(0), \ldots, \pi(31)$ and thus discloses the secret permutation $\pi$.

\par\smallskip\textbf{Using Matrix $A$.} 
We now present a second attack that only requires access to the hard-coded values $\vec A$, $\alpha$, $\beta$ defining the two matrices $\vec A$ and $\vec M$.
We can directly read off $\vec A$ from Algorithm~\ref{alg:whitebox:whitebox_pre}. 
Using Theorem~\ref{theo:whitebox:main}, we can subsequently compute   \[
      \mathbf{AM} = 
  \left(
  \begin{array}{cl}
  r_{\pi(0)} \vec{e_{\pi(0)}} & t_0 \\
  \vdots & \vdots  \\
  r_{\pi(31)} \vec{e_{\pi(31)}} & t_{31} \\  
  \end{array}
  \right),
  \]
which reveals the permuted randomization $\pi(\vec r) = (r_{\pi(0)}, \ldots , r_{\pi(31)})$ and the secret permutation $\pi$.

\par\smallskip\textbf{Private Key Recovery.} 
\textsc{Obfusc-Precomp}($c$) computes $\hat{\vec c}$ on ciphertext $c$ as input. 
Recall from Equation~(\ref{eq:whitebox:dec_correctness}) that the RSA decryption can be computed on obfuscated vector $\hat{\vec c}$ as 
\[ 
    m = c^d = \textsc{Exp}(\hat{ \textbf{c}}, \hat{\textbf{d}}) \cdot \prod_{i=0}^k r_{d_i}^{-32^i},
\]
using the obfuscated secret key $\hat{\mathbf{d}} = (\pi^{-1}(d_0), \ldots, \pi^{-1}(d_k))$. 
Consequently, we can extract $\hat{\mathbf{d}}$ from $\textsc{Exp}(\hat{ \textbf{c}}, \hat{\textbf{d}})$ and recover the original RSA secret key by computing 
\[
    \pi(\hat{\mathbf{d}}) = \vec d = (d_0, \ldots , d_k) \textrm{ and } d = \sum_{i=0}^k d_i \cdot 32^{i}.
\]

Given public key $e$ and private key $d$, we retrieve the factorization $N=pq$ using the approach described in Section~\ref{subsec:whitebox:rsa}.
}
\section{Discussion}\label{sec:whitebox:discussion}
In the following, we discuss implications of our industry-wide security analysis, reflect on potential defenses, elaborate on existing literature, and finally outline future research directions.

\subsection{Implications}\label{subsec:whitebox:implications}
We now discuss implications of our attacks from different perspectives.

\textbf{\ac{EDA} Tools.}
Although with varying complexity, each case study results in a full break of the confidentiality and integrity of the respective IEEE~1735 implementation.
In particular, we extract all RSA private keys used within the analyzed \ac{EDA} tools for the decryption of protected \ac{IP}.
As explicitly permitted by IEEE~1735, three out of the six analyzed vendors exhibit no (noteworthy) private key protections at all. 
Reverse engineering of these unprotected tools is sometimes even less time-consuming than their installation process.
Since most \ac{IP} is protected using at least one of the targeted vendors' RSA public keys, we can decrypt almost \textit{all} protected \acp{IP} used throughout the industry.
Consequently, the affected \ac{IP} does not only present a target for \ac{IP} theft, but is also susceptible to malicious manipulations such as hardware Trojans~\cite{Bhasin2013,Zhang2011}.

\par\smallskip\textbf{IEEE Standard 1735-2014.}
Since the execution environment of \ac{EDA} tools is typically untrusted, they are naturally susceptible to software reverse-engineering and manipulation.
In this context, the recommendation of hard-coded vendor keys within IEEE~1735 is inherently insecure. 
Countering this critical threat by declaring it \enquote{\textit{out-of-scope}}~\cite[p.~11]{Automation2015} within IEEE~1735 and encouraging the prohibition of reverse engineering by adoption of end-user license agreements~\cite[p.~11]{Automation2015}, rather than providing guidance for a secured implementation, demonstrates a lack of a realistic risk assessment.
Moreover, even if the \ac{IP} decryption could be performed securely, the decrypted \ac{IP} core would still be stored in \ac{RAM} at runtime and could thus be extracted. 
In this regard, the standard mentions that \enquote{\textit{Each tool vendor needs to decide what, if any, anti-debug, anti-tamper, and anti-key-discovery technologies they will use here.}}~\cite[p.~11]{Automation2015}.
However, the standard provides a false sense of security since in practice its soundness is entirely founded on the protection of the private keys and the plaintext \ac{IP}.
As in other \ac{DRM} applications, \textit{security-by-obscurity} constitutes the \textit{only} available protection layer in this case.
Additionally, an \ac{IP} may contain session keys encrypted for multiple different tool vendors, therefore a single unprotected private key within one \ac{EDA} tool can compromise the entire scheme.
Hence, the central security assumptions of IEEE~1735 are invalid in practice. 
Future iterations of IEEE~1735 should at least \textit{explicitly} emphasize the importance of software protection, outline limitations of cryptography in a real-world setting, and \textit{clearly} communicate the threats arising thereof. 

\par\smallskip\textbf{Public-Key White-Box Cryptography.}
While symmetric white-box cryptography has received significant attention in recent years~\cite{Billet2003,Bock2020,Bos2016,Chow2002a,Karroumi2010}, public-key white-box schemes have been analyzed scarcely in the open literature so far. 
However, (presumably insecure) proprietary public-key white-box schemes are already deployed in practice to protect high-value \ac{IP}. 
In its current state and due to limited public scrutiny, we claim that public-key white-box cryptography lacks the security assurances required to be deployed on its own, i.e., without any additional software protections.
Our case studies demonstrate proprietary cryptography to (again) be ineffective in practice. 
On the basis of our recovery of two commercial public-key white-box decryption schemes, we demonstrated its weaknesses in regard to both central security notions of white-box cryptography.

\subsection{Thoughts on IP Protection in Practice}\label{subsec:whitebox:though_experiment}
We now reflect on general implications and the potential mitigation of threats arising from untrusted execution environments used for \ac{IP} protection schemes.

\par\smallskip\textbf{Integrity and Authentication.}
Some use-cases demand integrity and authentication of the \ac{IP} rather than confidentiality, as most protected \ac{IP} implement well-known interfaces or algorithms that typically do not provide a significant edge over competitor product portfolios.
This is especially true for security-sensitive and safety-critical applications, f.i., in military or critical infrastructure systems.
In these use-cases, the ability to inspect third-party \ac{IP} before integration may even be required.
Such applications primarily benefit from tamper-resistance (e.g., enforced by integrity protections) and authenticity of the \ac{IP} vendor to prevent supply chain attacks.

Integrity and authentication of an \ac{IP} can be achieved using digital signatures in combination with a \ac{PKI} or (in offline-settings) by manual comparison of hash-values computed over the encrypted \ac{IP}. 
Since integrity and authenticity are primarily in the interest of the \ac{IP} user, there is typically no incentive for them to bypass these protections on their machines. 
Nonetheless, we want to emphasize that this cannot defend against software manipulations of the \ac{EDA} tool that would enable an attacker to interfere with the \ac{IP} after integrity checks have been performed.
As a result, cryptographic integrity protections can defend against \ac{MITM} attackers trying to intercept and manipulate the \ac{IP} during shipment, but still fail to prevent attacks within the \ac{MATE} attacker model.

\par\smallskip\textbf{Confidentiality.}
Introducing a \ac{PKI} for the management and distribution of individual session keys at first seems to be an obvious choice.
However, the private keys will still end up in the \ac{EDA} tool's program memory and are therefore vulnerable to extraction.
Hence, for the sake of our argument, let us assume that a perfectly secure key storage solution exists and all cryptographic operations can be performed in a trusted execution environment (e.g., Intel SGX or ARM TrustZone).
Following our attacker model, a reverse engineer would still be able to tamper with the design software itself and read out its memory at runtime.
As soon as the \ac{EDA} tool performs any kind of operation on the protected \ac{IP} core (e.g., run a synthesis or implementation step), it must decrypt the \ac{IP} and keep its plaintext representation in memory during the entire operation.
At this point, an attacker can recover the decrypted \ac{IP} by dumping and inspecting the memory during operation.
To counter this threat, the \textit{entire} \ac{EDA} tool would have to  run in a trusted execution environment, which is not desirable for performance reasons. 

Even if all \ac{EDA} tool computations were to be performed in such a trusted environment, the result of the hardware design flow is either a bitstream file that configures an \ac{FPGA} or a layout description for an \ac{ASIC}.
Crucially, both a bitstream~\cite{SymbiFlow2018,Wolf,Yu2018} and a layout~\cite{Rajarathnam2020} contain all  information about an \ac{IP} core as they are simply a different gate-level netlist representation. 
A determined adversary can extract the high-level functionality of an \ac{IP} core through netlist analysis~ \cite{Albartus2020,MPI2019,Meade2018a,Subramanyan2014a}. 
We note that even though this approach arguably requires considerable efforts and specialized knowledge on hardware reverse engineering, it always succeeds given a motivated adversary.
Note that this is similar to other \ac{DRM} settings, as the protected data must be handed to the user eventually.

In summary, perfect confidentiality of protected \ac{IP} cannot be achieved, even when using strong cryptographic solutions and tools that are running in a trusted execution environment.
Thus so far, the only viable defenses are raise-the-bar countermeasures based on strong software defense mechanisms.

\par\smallskip\textbf{Raise-the-Bar Countermeasures.}
Similar to \ac{DRM} use-cases, cryptography alone cannot resolve the aforementioned issues to achieve confidentiality, hence the goal must be to increase the economic effort such that the costs of an infringement outweigh any potential profit. 
In our case studies, potent software obfuscation proves to be the one countermeasure that significantly increases attack time, while most anti-debugging measures were defeated using publicly available tooling.
Cryptography and even white-box implementations may play a vital part, but, in their current state must always be combined with strong software obfuscation techniques to suppress any attempt on static or dynamic analysis.

Another approach to accomplish confidentiality in such untrusted execution environments may be to outsource all computations that require access to a decrypted \ac{IP} core to a trusted (cloud-)server and employ a \textit{thin-client} solution.
This mitigates the threat of \ac{IP} disclosure and manipulation to some extend as the protected \ac{IP} itself is never decrypted on the \ac{IP} user workstation.
While thin-client solutions may be a viable for some use-cases, others such as the hardware development for military or critical infrastructure applications require the entire design process to be performed on \textit{air-gapped} machines located on manufacturer premises.
Despite all efforts, the \ac{IP} user (and potential attacker) still receives a bitstream or layout in the end, which may then be targeted by netlist reverse-engineering to defeat confidentiality and integrity.

\subsection{Related Work}\label{subsec:whitebox:related_work}
Our work relates to aspects of both hardware \ac{IP} protection and white-box cryptography.

\par\smallskip\textbf{IP Protection.}
The threat of hardware \ac{IP} infringement and the development of techniques to achieve \ac{IP} protection is widely covered in the literature~\cite{Bossuet2017,Mishra2017}.
While works commonly focus on hardware obfuscation, IEEE~1735~\cite{Automation2015} presents the first approach to an industry-wide solution based on well-known cryptographic algorithms.
In 2017, Chhotaray et~al.~\cite{Chhotaray2017} identify flaws in the choice of cryptographic primitives of IEEE~1735.
They mount both a padding-oracle and a syntax-oracle attack on an \ac{EDA} tool for \acp{FPGA}.
Subsequently, they recover the plaintext \ac{IP} and apply meaningful manipulations to insert a hardware Trojan.
However, their attacks did not question the fundamental security assumptions of IEEE~1735.
In response to findings of Chhotaray et~al., the IEEE issued a statement asking \ac{IP} owners and users to secure their supply chain and implement end-user agreements~\cite{Statement2017}.
The standard itself remains unchanged to this day.
Mirian et~al.~\cite{Mirian2016} discuss flaws within the implementation of Xilinx ISE, which was discontinued in 2012.
They extract the gate-level netlist of an \ac{IP} protected by IEEE~1735 by exploiting the insecure design flow of Xilinx ISE, but do not recover a high-level description.

\par\smallskip\textbf{White-Box Cryptography.}
White-box cryptography was first introduced by Chow et~al.~\cite{Chow2002a, Chow2002} in 2002.
In the last decades numerous efforts to formalize symmetric white-box security properties have been published~\cite{Bock2020,Delerablee2013,Saxena2009}. 
For now, white-box cryptography resembles an arms race with many such propositions \cite{Chow2002a,Chow2002,Karroumi2010} being broken shortly after publication \cite{Billet2003,Billet2004,Bos2016}.
In 2017, 2019, and 2021, white-box competitions were held at CHES~\cite{CHES2017,CHES2019,CHES2021} to stimulate research of white-box implementations and their security. 
The first competition was subsequently evaluated by Bock et~al.~\cite{Bock2020b}.
Apart from classical attacks, gray-box approaches that use side-channel information to extract information from the white-box present another critical threat~\cite{Bos2016,Goubin2020,Sanfelix2015}.  
In academia, white-box research is typically focused on symmetric encryption schemes.
To the best of our knowledge, the work of Barthelemy~\cite{Barthelemy2020} is the only (academic) publication on public-key white-box cryptography based on a novel lattice-based approach.
The latest white-box competition at CHES~\cite{CHES2021} set out the target to develop a secure ECDSA white-box implementation.
At the time of writing, each and every contestant has been broken within two days of publication.

\subsection{Future Work}\label{subsec:whitebox:future_work}
In light of our discussion in Section~\ref{subsec:whitebox:though_experiment}, we argue that an all-encompassing solution for \ac{IP} protection is out of reach given the contemporary hardware design process.
Nonetheless, IEEE standard 1735-2014 is in dire need for a timely revision that not only takes its choice of cryptographic primitives into account, but pays respect to its foundational security assumptions in untrusted execution environments as well.
Furthermore, additional precautions for integrity assurance of the \textit{entire} \ac{IP} are essential to restore trust in a future iteration of IEEE~1735.
Hence, we provide an incentive for standardization bodies and researchers alike to develop sound solutions for the protection of hardware \ac{IP}.

Our work practically demonstrates the discrepancy between commercially-available public-key white-box solutions and the limited understanding of public-key white-box security in academia. 
Hence, both academia and industry can build upon our work to take the presented cryptanalytic attacks into consideration and build the foundations for a secure usage of public-key white-box cryptography in the future. 
Moreover, we hope that our work sparks interest in developing (and analyzing) public-key white-box algorithms to improve the understanding of such schemes.
\section{Conclusion}\label{sec:whitebox:conclusion}
IEEE standard 1735-2014 aims to provide recommended practices for the secure management of electronic design \ac{IP}~\cite{Automation2015}. 
Our work presented structural weaknesses of IEEE~1735 that allow for extraction of the \ac{EDA} tool vendors' private keys. 
Consequently, our attacks enable to decrypt, maliciously modify, and re-encrypt all allegedly protected \ac{IP}. 
In particular, our case studies uncovered that three vendors use no software protection to safeguard their \textit{hard-coded} private keys, two vendors employ code obfuscation and anti-debugging, and three vendors utilize a public-key white-box cryptographic approach. 
We analyzed the white-box schemes and presented cryptanalytical attacks on all three schemes to effectively extract their hidden keys within seconds.

Our insights on realistic software analysis capabilities in untrusted execution environments demand a rethinking of the use of third-party \ac{IP} cores in its current form as both \ac{IP} authors and \ac{IP} users are at risk with limited prospect of improvement.
\printbibliography

@misc{Thales,
author = {Thales},
title = {{White Box Cryptography}},
url = {https://cpl.thalesgroup.com/software-monetization/white-box-cryptography},
urldate = {2021-12-09}
}

@misc{Digital.ai,
author = {Digital.ai},
title = {{White-Box Cryptography}},
url = {https://digital.ai/glossary/whitebox-cryptography},
urldate = {2021-12-09}
}

@misc{Irdeto,
author = {Irdeto},
title = {{Whitebox Cryptography}},
url = {https://irdeto.com/whitebox-cryptography/},
urldate = {2021-12-09}
}

@inproceedings{Chow2002a,
author = {Chow, Stanley and Eisen, Phil and Johnson, Harold and {Van Oorschot}, Paul C.},
booktitle = {ACM CCS-9 Workshop on Security and Privacy in Digital Rights Management (DRM)},
pages = {1--15},
publisher = {Springer},
title = {{A white-box DES implementation for DRM applications}},
url = {https://crypto.stanford.edu/DRM2002/whitebox.pdf},
year = {2002}
}

@book{Automation2015,
author = {{IEEE Design Automation Standards Committee (DASC)}},
publisher = {IEEE},
title = {{IEEE 1735-2014 - Recommended Practice for Encryption and Management of Electronic Design Intellectual Property (IP)}},
year = {2015}
}

@inproceedings{Chow2002,
author = {Chow, Stanley and Eisen, Philip and Johnson, Harold and {Van Oorschot}, Paul C.},
booktitle = {9th Annual International Conference on Selected Areas in Cryptography (SAC)},
number = {2595},
pages = {250--270},
publisher = {Springer},
title = {{White-box cryptography and an AES implementation}},
url = {https://home.cs.colorado.edu/~jrblack/class/csci7000/f03/papers/oorschot-whitebox.pdf},
year = {2002}
}

@article{Bock2020,
author = {Bock, Estuardo Alpirez and Amadori, Alessandro and Brzuska, Chris},
journal = {IACR Transactions on Cryptographic Hardware and Embedded Systems (TCHES)},
number = {2},
pages = {327--357},
title = {{On the Security Goals of White-Box Cryptography}},
url = {https://tches.iacr.org/index.php/TCHES/article/view/8554/8119},
year = {2020}
}

@inproceedings{Delerablee2013,
author = {Delerabl{\'{e}}e, C{\'{e}}cile and Lepoint, Tancr{\`{e}}de and Paillier, Pascal and Rivain, Matthieu},
booktitle = {20th Annual International Conference on Selected Areas in Cryptography (SAC)},
pages = {247--264},
publisher = {Springer},
title = {{White-Box Security Notions for Symmetric Encryption Schemes}},
url = {https://eprint.iacr.org/2013/523.pdf},
volume = {8282},
year = {2013}
}

@techreport{Krolikoski2020,
author = {Krolikoski, Stan},
institution = {IEEE Design Automation Standards Committee (DASC)},
title = {{Year End 2020 DASC Report}},
url = {https://www.dasc.org/DASC Report 11.2020.pdf},
year = {2020}
}

@inproceedings{Chhotaray2017,
author = {Chhotaray, Animesh and Nahiyan, Adib and Shrimpton, Thomas and Forte, Domenic and Tehranipoor, Mark},
booktitle = {Proceedings of the 2017 ACM SIGSAC Conference on Computer and Communications Security (CCS)},
pages = {1533--1546},
title = {{Standardizing Bad Cryptographic Practice: A Teardown of the IEEE Standard for Protecting Electronic-design Intellectual Property}},
url = {https://acmccs.github.io/papers/p1533-chhotarayA.pdf},
year = {2017}
}

@article{Montgomery1985,
author = {Montgomery, Peter L.},
journal = {Mathematics of Computation},
number = {170},
pages = {519},
title = {{Modular Multiplication Without Trial Division}},
url = {https://www.ams.org/journals/mcom/1985-44-170/S0025-5718-1985-0777282-X/S0025-5718-1985-0777282-X.pdf},
volume = {44},
year = {1985}
}

@patent{Hoogerbrugge2018,
author = {Hoogerbrugge, Jan and Michiels, Wil and Vullers, Pim},
publisher = {United States Patent and Trademark Office},
title = {{White-Box Elliptic Curve Point Multiplication}},
year = {2018},
number = {U.S. Patent 10,068,070 B2}
}

@patent{Hoogerbrugge2019,
author = {Hoogerbrugge, Jan and Michiels, Wil},
publisher = {United States Patent and Trademark Office},
title = {{White-Box Modular Exponentiation}},
year = {2019},
number={U.S. Patent 10,235,506 B2}
}

@patent{Michiels2013,
author = {Michiels, Wil and Gorissen, Paulus},
publisher = {United States Patent and Trademark Office},
title = {{Exponent Obfuscation}},
year = {2013},
number={U.S. Patent 8,600,047 B2}
}

@patent{Zhou2009,
author = {Zhou, Yongxin and Chow, Stanley T.},
publisher = {United States Patent and Trademark Office},
title = {{System and Method of Hiding Cryptographic Private Keys}},
year = {2009},
number={U.S. Patent 7,634,091 B2}
}

@patent{Hoogerbrugge2020,
author = {Hoogerbrugge, Jan and Michiels, Wil},
publisher = {United States Patent and Trademark Office},
title = {{Protecting the Input/Output of Modular Encoded White-Box RSA}},
year = {2020},
number={U.S. Patent 10,726,108 B2}
}

@article{Barthelemy2020,
author = {Barthelemy, Lucas},
journal = {IACR Cryptology ePrint Archive},
pages = {893},
title = {{Toward an Asymmetric White-Box Proposal}},
url = {https://eprint.iacr.org/2020/893.pdf},
volume = {2020},
year = {2020}
}

@article{Sanfelix2015,
author = {Sanfelix, Eloi and Mune, Cristofaro and de Haas, Job},
journal = {Proceedings of the 2015 Black Hat Europe Conference},
title = {{Unboxing the White-Box: Practical attacks against Obfuscated Ciphers}},
url = {https://www.blackhat.com/docs/eu-15/materials/eu-15-Sanfelix-Unboxing-The-White-Box-Practical-Attacks-Against-Obfuscated-Ciphers-wp.pdf},
year = {2015}
}

@article{Mirian2016,
author = {Mirian, Vincent and Chow, Paul},
journal = {Proceedings of the ACM Great Lakes Symposium on VLSI (GLSVLSI)},
pages = {293--298},
title = {{Extracting designs of secure IPs using FPGA CAD tools}},
year = {2016}
}

@article{Goubin2020,
author = {Goubin, Louis and Rivain, Matthieu and Wang, Junwei},
journal = {IACR Transactions on Cryptographic Hardware and Embedded Systems (TCHES)},
number = {3},
pages = {454--482},
title = {{Defeating State-of-the-Art White-Box Countermeasures with Advanced Gray-Box Attacks}},
url = {https://tches.iacr.org/index.php/TCHES/article/view/8597/8164},
year = {2020}
}

@inproceedings{Bos2016,
author = {Bos, Joppe W. and Hubain, Charles and Michiels, Wil and Teuwen, Philippe},
booktitle = {18th International Conference on Cryptographic Hardware and Embedded Systems (CHES)},
pages = {215--236},
publisher = {Springer},
title = {{Differential computation analysis: Hiding your white-box designs is not enough}},
url = {https://www.iacr.org/archive/ches2016/98130106/98130106.pdf},
year = {2016}
}

@article{Bock2020b,
author = {Bock, Estuardo Alpirez and Treff, Alexander},
journal = {IACR Cryptology ePrint Archive},
pages = {342},
title = {{Security Assessment of White-Box Design Submissions of the CHES 2017 CTF Challenge}},
url = {https://eprint.iacr.org/2020/342.pdf},
volume = {2020},
year = {2020}
}

@inproceedings{AlpirezBock2020,
author = {{Alpirez Bock}, Estuardo and Brzuska, Chris and Fischlin, Marc and Janson, Christian and Michiels, Wil},
booktitle = {26th International Conference on the Theory and Application of Cryptology and Information Security (ASIACRYPT)},
pages = {221--252},
publisher = {Springer},
title = {{Security Reductions for White-Box Key-Storage in Mobile Payments}},
url = {https://eprint.iacr.org/2019/1014.pdf},
volume = {12491},
year = {2020}
}

@article{Miller1976,
author = {Miller, Gary L.},
journal = {Journal of Computer and System Sciences},
number = {3},
pages = {300--317},
title = {{Riemann's Hypothesis and Tests for Primality}},
url = {https://dl.acm.org/doi/abs/10.1145/800116.803773},
volume = {13},
year = {1976}
}

@inproceedings{Saxena2009,
author = {Saxena, Amitabh and Wyseur, Brecht and Preneel, Bart},
booktitle = {12th International Conference on Information Security (ISC)},
pages = {49--58},
publisher = {Springer},
title = {{Towards security notions for white-box cryptography}},
url = {https://www.esat.kuleuven.be/cosic/publications/article-1260.pdf},
volume = {5735},
year = {2009}
}

@techreport{Statement2017,
author = {{IEEE Design Automation Standards Committee (DASC)}},
institution = {IEEE},
title = {{Security Statement}},
url = {http://standards.ieee.org/about/ieee_1735.pdf},
year = {2017}
}

@article{Albartus2020,
author = {Albartus, Nils and Hoffmann, Max and Temme, Sebastian and Azriel, Leonid and Paar, Christof},
journal = {IACR Transactions on Cryptographic Hardware and Embedded Systems (TCHES)},
number = {4},
pages = {309--336},
title = {{DANA Universal Dataflow Analysis for Gate-Level Netlist Reverse Engineering}},
url = {https://tches.iacr.org/index.php/TCHES/article/view/8685/8244},
year = {2020}
}

@article{Meade2018a,
author = {Meade, Travis and Shamsi, Kaveh and Le, Thao and Di, Jia and Zhang, Shaojie and Jin, Yier},
journal = {Journal of Hardware and Systems Security},
number = {3},
pages = {201--213},
publisher = {Journal of Hardware and Systems Security},
title = {{The Old Frontier of Reverse Engineering: Netlist Partitioning}},
url = {http://jin.ece.ufl.edu/papers/HASS2018_netlist.pdf},
volume = {2},
year = {2018}
}

@book{Mishra2017,
author = {Mishra, Prabhat and Bhunia, Swarup and Tehranipoor, Mark},
edition = {1},
publisher = {Springer},
title = {{Hardware IP Security and Trust}},
year = {2017}
}

@book{Bossuet2017,
author = {Bossuet, Lilian and Torres, Lionel},
edition = {1},
publisher = {Springer},
title = {{Foundations of Hardware IP Protection}},
year = {2017}
}

@article{Akhunzada2015,
author = {Akhunzada, Adnan and Sookhak, Mehdi and Anuar, Nor Badrul and Gani, Abdullah and Ahmed, Ejaz and Shiraz, Muhammad and Furnell, Steven and Hayat, Amir and {Khurram Khan}, Muhammad},
journal = {Journal of Network and Computer Applications},
pages = {44--57},
publisher = {Elsevier},
title = {{Man-At-The-End attacks: Analysis, taxonomy, human aspects, motivation and future directions}},
volume = {48},
year = {2015}
}

@inproceedings{Rajarathnam2020,
author = {Rajarathnam, Rachel Selina and Lin, Yibo and Jin, Yier and Pan, David Z.},
booktitle = {IEEE International Symposium on Hardware Oriented Security and Trust (HOST)},
pages = {154--163},
title = {{ReGDS: A Reverse Engineering Framework from GDSII to Gate-level Netlist}},
url = {https://ieeexplore.ieee.org/abstract/document/9300272},
year = {2020}
}

@misc{MPI2019,
author = {{Max Planck Institute for Security and Privacy}},
title = {{HAL - The Hardware Analyzer}},
url = {https://github.com/emsec/hal},
year = {2019},
urldate = {2021-12-09}
}

@article{Subramanyan2014a,
author = {Subramanyan, Pramod and Tsiskaridze, Nestan and Li, Wenchao and Gasc{\'{o}}n, Adri{\`{a}} and Tan, Wei Yang and Tiwari, Ashish and Shankar, Natarajan and Seshia, Sanjit A. and Malik, Sharad},
journal = {IEEE Transactions on Emerging Topics in Computing},
number = {1},
pages = {63--80},
title = {{Reverse engineering digital circuits using structural and functional analyses}},
volume = {2},
url = {https://sites.cs.ucsb.edu/~nestan/pdf/TETCSI.pdf},
year = {2014}
}

@article{Yu2018,
author = {Yu, Hoyoung and Lee, Hansol and Lee, Sangil and Kim, Youngmin and Lee, Hyung Min},
journal = {Electronics},
number = {10},
pages = {1--14},
title = {{Recent advances in FPGA reverse engineering}},
url = {https://www.mdpi.com/2079-9292/7/10/246},
volume = {7},
year = {2018}
}

@book{Bhunia2017,
author = {Bhunia, Swarup and Ray, Sandip and Sur-Kolay, Susmita},
edition = {1st},
publisher = {Springer},
title = {{Fundamentals of IP and SoC Security: Design, Verification, and Debug}},
year = {2017}
}

@misc{CHES2021,
author = {{Conference on Cryptographic Hardware and Embedded Systems (CHES)}},
title = {{Whib0x Contest 2021}},
%year = {2021}
}

@misc{CHES2019,
author = {{Conference on Cryptographic Hardware and Embedded Systems (CHES)}},
title = {{Whib0x Contest 2019}},
%year = {2019}
}

@misc{CHES2017,
author = {{Conference on Cryptographic Hardware and Embedded Systems (CHES)}},
title = {{Whib0x Contest 2017}},
%year = {2017}
}

@article{Bhasin2013,
author = {Bhasin, Shivam and Danger, Jean Luc and Guilley, Sylvain and Ngo, Xuan Thuy and Sauvage, Laurent},
journal = {10th Workshop on Fault Diagnosis and Tolerance in Cryptography (FDTC)},
pages = {15--29},
publisher = {IEEE},
title = {{Hardware trojan horses in cryptographic IP cores}},
url = {https://eprint.iacr.org/2014/750.pdf},
year = {2013}
}

@article{Zhang2011,
author = {Zhang, Xuehui and Tehranipoor, Mohammad},
journal = {IEEE International Symposium on Hardware-Oriented Security and Trust (HOST)},
pages = {67--70},
publisher = {IEEE},
title = {{Case study: Detecting hardware Trojans in third-party digital IP cores}},
year = {2011}
}

@article{Billet2004,
author = {Billet, Olivier and Gilbert, Henri and Ech-Chatbi, Charaf},
journal = {11th International Workshop on Selected Areas in Cryptography (SAC)},
pages = {227--240},
title = {{Cryptanalysis of a White Box AES Implementation}},
url = {https://bo.blackowl.org/s/papers/waes.pdf},
volume = {3357},
year = {2004}
}

@inproceedings{Karroumi2010,
author = {Karroumi, Mohamed},
booktitle = {13th International Conference on Information Security and Cryptology (ICISC)},
pages = {278--291},
publisher = {Springer},
title = {{Protecting white-box AES with dual ciphers}},
volume = {6829},
year = {2010}
}

@inproceedings{Billet2003,
author = {Billet, Olivier and Gilbert, Henri},
booktitle = {9th International Conference on the Theory and Application of Cryptology and Information Security (ASIACRYPT)},
pages = {331--346},
title = {{A traceable block cipher}},
url = {https://www.iacr.org/cryptodb/archive/2003/ASIACRYPT/31/31.pdf},
volume = {2894},
year = {2003}
}

@article{Chen2016,
author = {Chen, Ping and Huygens, Christophe and Desmet, Lieven and Joosen, Wouter},
journal = {IFIP International Conference on ICT Systems Security and Privacy Protection},
pages = {323--336},
title = {{Advanced or not? A comparative study of the use of anti-debugging and anti-VM techniques in generic and targeted malware}},
volume = {471},
url = {https://core.ac.uk/download/pdf/34612995.pdf},
year = {2016}
}

@inproceedings{Chen2008,
author = {Chen, Xu and Andersen, Jonathon and Mao, Zhuoqing Morley and Bailey, Michael and Nazario, Jose},
booktitle = {38th Annual IEEE/IFIP International Conference on Dependable Systems and Networks (DSN)},
pages = {177--186},
publisher = {IEEE},
title = {{Towards an Understanding of Anti-virtualization and Anti-debugging Behavior in Modern Malware}},
url = {https://web.eecs.umich.edu/~zmao/Papers/DCCS-xu-chen.pdf},
year = {2008}
}

@techreport{Collberg1997,
author = {Collberg, Christian and Thomborson, Clark and Low, Douglas},
booktitle = {The University of Auckland},
title = {{A taxonomy of obfuscating transformations}},
year = {1997}
}

@phdthesis{Distel2017,
author = {Distel, Dominic},
school = {Technical University of Munich},
title = {{Markets for Technology in the Semiconductor Industry – The Role of Ability-Related Trust in the Market for IP Cores}},
url = {https://mediatum.ub.tum.de/doc/1357158/1357158.pdf},
year = {2017}
}

@manual{XilinxIP,
author = {{Xilinx}},
title = {{Vivado Design Suite User Guide -- Creating and Packaging Custom IP (UG1118)}},
version = {v2021.1},
url = {https://www.xilinx.com/support/documentation/sw_manuals/xilinx2021_1/ug1118-vivado-creating-packaging-custom-ip.pdf},
date = {2021-06-30}
}

@manual{CadenceIP,
author = {{Cadence Design Systems}},
title = {{Using the IEEE 1735 protection mechanism with a Public key to protect Verilog code or VHDL code, and how models can share between vendor tool sets, DECERR or CORRPD error}},
date = {2021-08-11}
}

@manual{IntelIP,
author = {{Intel}},
title = {{Introduction to Intel® FPGA IP Cores}},
version = {2020.11.09},
url = {https://www.intel.com/content/dam/www/programmable/us/en/pdfs/literature/ug/ug_intro_to_megafunctions.pdf},
date = {2020-09-11}
}

@manual{LatticeIP,
author = {{Lattice Semiconductor}},
title = {{Lattice Radiant Software 3.0 User Guide}},
version = {3.0},
url = {https://www.latticesemi.com/view_document?document_id=53229},
date = {2021-06-14}
}

@manual{MicrosemiIP,
author = {{Microchip Technology}},
title = {{Secure IP Flow for IP Vendors and Libero SoC Users}},
version = {A},
url = {https://www.microsemi.com/document-portal/doc_download/133573-libero-soc-secure-ip-flow-user-guide},
date = {2021-05}
}

@manual{SiemensIP,
author = {{Siemens}},
title = {{ModelSim® User's Manual}},
version = {v10.5c},
url = {https://www.microsemi.com/document-portal/doc_download/136662-modelsim-me-10-5c-user-u-s-manual-for-libero-soc-v11-8},
}

@misc{SymbiFlow2018,
author = {SymbiFlow Team},
title = {{Project X-Ray}},
url = {https://github.com/SymbiFlow/prjxray},
year = {2021},
urldate = {2021-12-09}
}

@misc{Wolf,
author = {Wolf, Claire and Lasser, Mathias},
title = {{Project IceStorm}},
url = {http://bygone.clairexen.net/icestorm/},
year = {2018},
urldate = {2021-12-09}
}

@misc{scyllahyde,
author = {Carbon and cypher and mrexodia and Mattiwatti},
title = {ScyllaHide},
url = {https://github.com/x64dbg/ScyllaHide},
year = {2021},
urldate = {2021-12-09}
}

@misc{widevine,
author = {tomer8007},
title = {Widewine L3 Decryptor},
url = {https://github.com/tomer8007/widevine-l3-decryptor},
year = {2020},
urldate = {2021-12-09}
}

@misc{IntelAck,
author = {{Intel}},
title = {{Intel Quartus Prime Pro Edition: Version 21.3 Software and Device Support Release Notes}},
url = {https://www.intel.com/content/www/us/en/programmable/documentation/ewa1443722509979.html},
year = {2021}
}

@misc{SiemensAck,
author = {{Siemens}},
title = {{SSA-400332: Insufficient Design IP Protection in Questa and ModelSim}},
url = {https://cert-portal.siemens.com/productcert/pdf/ssa-400332.pdf},
year = {2021},
note = {(available latest on 12/14/2021)}
}

@misc{XilinxAck,
author = {{Xilinx}},
title = {{Findings on IEEE-1735-2014 recommended practice from Max Planck Institute for Security and Privacy (MPI-SP) and Ruhr University Bochum (RUB)}},
url = {https://support.xilinx.com/s/article/000033507},
year = {2021}
}

\appendix[IEEE and Vendor Responses]
The vendors covered by our case studies have been cooperative throughout the entire responsible disclosure process and (at the point of writing) have either updated their tools already or are working towards fixing the disclosed vulnerabilities. 
Moreover, several affected vendors publicly acknowledge our findings in coordination with our publication \intel{\siemens{\cite{IntelAck,SiemensAck,XilinxAck}}} and guide users through the necessary steps to improve the protection of their \ac{IP}.
So far, the IEEE has declined multiple offers to participate in discussions on our attacks and the future directions of IEEE~1735.
We note that within their correspondence, the IEEE insisted on referring to IEEE~1735 as a \enquote{recommended practice} instead of a \enquote{standard}.

\end{document}